%% file: Main.tex
\RequirePackage[l2tabu,orthodox]{nag}		
\documentclass[reqno]{amsart}

\usepackage{amsmath}		
\usepackage{amssymb}		
\usepackage{amsfonts}		
\usepackage{amsthm}		
\usepackage[foot]{amsaddr}		

\usepackage{mathtools}		
\mathtoolsset{%
}

\usepackage[latin9]{inputenc}
\usepackage[T1]{fontenc}		

\usepackage[proportional,tabular,lining,sf,mono=false]{libertine}

\usepackage[labelfont={bf,small},labelsep=colon,font=small]{caption}	
\usepackage[dvipsnames,svgnames]{xcolor}		
\colorlet{MyBlue}{DodgerBlue!75!Black}
\colorlet{MyGreen}{DarkGreen!85!Black}

\newcommand{\Paragraph}[1]{\medskip\paragraph{\textbf{#1}}}

\usepackage{tikz}		
\usetikzlibrary{calc,patterns}
\newcommand\tikzmark[2]{%
\tikz[remember picture,baseline] \node[above, outer sep=0pt, inner sep=0pt] (#1){\phantom{#2}};%
}
\newcommand\link[2]{%
\begin{tikzpicture}[remember picture, overlay, >=stealth, shift={(0,0)}]
  \draw[->, line width=0.7mm] (#1) to (#2);
\end{tikzpicture}%
}

\usepackage{array}
\usepackage{booktabs}		
\usepackage{paralist}		

\usepackage{microtype}		

\usepackage{acronym}		
\usepackage{xspace}		

\usepackage[authoryear,sort&compress]{natbib}		

\def\bibsep{\smallskipamount}

\usepackage{hyperref}
\hypersetup{
colorlinks=true,
linktocpage=true,
pdfstartview=FitH,
breaklinks=true,
pdfpagemode=UseNone,
pageanchor=true,
pdfpagemode=UseOutlines,
plainpages=false,
bookmarksnumbered,
bookmarksopen=false,
bookmarksopenlevel=1,
hypertexnames=true,
pdfhighlight=/O,
urlcolor=Maroon,linkcolor=MyBlue!60!black,citecolor=DarkGreen!70!black,	
pdftitle={},
pdfauthor={},
pdfsubject={},
pdfkeywords={},
pdfcreator={pdfLaTeX},
pdfproducer={LaTeX with hyperref}
}

\usepackage[sort&compress,capitalize,nameinlink]{cleveref}		

\newcommand{\debug}[1]{#1}		

\theoremstyle{plain}
\newtheorem{theorem}{Theorem}		
\newtheorem{lemma}[theorem]{Lemma}		
\newtheorem{proposition}[theorem]{Proposition}		

\newtheorem*{corollary*}{Corollary}		

\theoremstyle{definition}
\newtheorem{definition}[theorem]{Definition}		

\newtheorem*{definition*}{Definition}		
\newtheorem*{assumption*}{Assumptions}		

\theoremstyle{remark}

\newtheorem*{remark*}{Remark}		
\newtheorem*{example*}{Example}		

\newcounter{proofpart}
\newenvironment{proofpart}[1]
{\vspace{3pt}
\refstepcounter{proofpart}%
\par\textit{Part~\arabic{proofpart}:~#1}.\,}
{\smallskip}

\DeclarePairedDelimiter{\bracks}{[}{]}		
\DeclarePairedDelimiter{\parens}{(}{)}		

\DeclarePairedDelimiter{\abs}{\lvert}{\rvert}		
\DeclarePairedDelimiter{\ceil}{\lceil}{\rceil}		
\DeclarePairedDelimiter{\floor}{\lfloor}{\rfloor}		

\DeclarePairedDelimiterX{\setdef}[2]{\{}{\}}{#1:#2}		
\DeclarePairedDelimiterXPP{\exclude}[1]{\mathopen{}\setminus}{\{}{\}}{}{#1}

\newcommand{\cf}{cf.\xspace}		
\newcommand{\eg}{e.g.,\xspace}		
\newcommand{\ie}{i.e.,\xspace}		
\newcommand{\textpar}[1]{\textup(#1\textup)}		

\DeclareMathOperator{\bigoh}{\mathcal O}		

\newcommand{\from}{\colon}		

\newcommand{\dd}{\:d}		

\DeclareMathOperator{\ex}{\mathbb{E}}		
\DeclareMathOperator{\prob}{\mathbb{P}}		
\DeclareMathOperator{\Var}{Var}		

\providecommand\given{}		

\DeclarePairedDelimiterXPP{\exof}[1]{\ex}{[}{]}{}{
\renewcommand\given{\nonscript\:\delimsize\vert\nonscript\:\mathopen{}} #1}

\DeclarePairedDelimiterXPP{\probof}[1]{\prob}{(}{)}{}{
\renewcommand\given{\nonscript\:\delimsize\vert\nonscript\:\mathopen{}} #1}

\renewcommand*{\arraystretch}{1.8}

\makeatletter

\newcommand{\Rmnum}[1]{\expandafter\@slowromancap\romannumeral #1@}
\makeatother

\DeclareMathOperator{\cost}{cost}
\DeclareMathOperator{\wait}{wait}
\newcommand{\g}{g}

\newcommand{\boys}{\mathcal{\debug C}}
\newcommand{\nBoys}{\debug N_{\boys}}
\newcommand{\NBoys}{\debug M_{\boys}}

\newcommand{\boy}{\debug i}

\newcommand{\girls}{\mathcal{\debug P}}
\newcommand{\nGirls}{\debug N_{\girls}}
\newcommand{\NGirls}{\debug M_{\girls}}
\newcommand{\girl}{\debug j}

\newcommand{\nShort}{\debug N}
\newcommand{\match}{\debug w}
\newcommand{\rate}{\debug \lambda}

\newcommand{\horizon}{\debug T}

\newcommand{\CS}{\ensuremath{\mathtt{CS}}\xspace}
	
\newcommand{\balanced}{\mathtt{balanced}}
\newcommand{\patient}{\mathtt{patient}}
\newcommand{\greedy}{\mathtt{greedy}}
\newcommand{\FIFO}{\mathtt{FCFS}}

\begin{document}




\title
[Breaking points in dynamic markets]
{Quick or cheap?\\
Breaking points in dynamic markets$^{\ast}$}

\author
[P.~Mertikopoulos]
{Panayotis Mertikopoulos$^{\diamond}$}
\email{panayotis.mertikopoulos@imag.fr}

\author
[H.~H.~Nax]
{Heinrich H.~Nax$^{\S}$}
\email{heinrich.nax@uzh.ch}

\author
[B.~S.~R.~Pradelski]
{Bary~S.~R.~Pradelski$^{\diamond}$}
\email{bary.pradelski@cnrs.fr}

\address{$^{\ast}$
\normalfont
The paper has benefited from comments by Vahideh Manshadi, Igal Milchtaich, Jonathan Newton, Sven Seuken, Philipp Strack and seminar participants at Bar-Ilan University, the Paris Game Theory Seminar, and the INFORMS Workshop on Market Design 2019. We thank Simon Jantschgi and Dimitrios Moustakas for careful proof-reading. All errors are ours. PM benefited from the support of the COST Action CA16228 ``European Network for Game Theory'' (GAMENET). HN benefited from the support of the ERC Advanced Investigator Grant `Momentum' (No. 324247).
BP benefited from the support of the ANR grants 15-IDEX-02, 11LABX0025-01, ALIAS and the Oxford-Man Institute of Quantitative Finance.}

\address{$^{\diamond}$ Univ. Grenoble Alpes, CNRS, Inria, Grenoble INP, LIG, 38000 Grenoble, France}
\address{$^{\S}$ Univ. Zurich, 8050 Zurich, Switzerland}



\newcommand{\acli}[1]{\textit{\acl{#1}}}		
\newcommand{\aclip}[1]{\textit{\aclp{#1}}}		
\newcommand{\acdef}[1]{\textit{\acl{#1}} \textup{(\acs{#1})}\acused{#1}}		
\newcommand{\acdefp}[1]{\emph{\aclp{#1}} \textup(\acsp{#1}\textup)\acused{#1}}	

\newacro{ME}{matching event}
\newacro{FCFS}{first-come, first-served}
\newacro{CS}{clearing schedule}
\newacro{LHS}{left-hand side}
\newacro{RHS}{right-hand side}
\newacro{iid}[i.i.d.]{independent and identically distributed}
\newacro{NE}{Nash equilibrium}
\newacroplural{NE}[NE]{Nash equilibria}

\begin{abstract}
\input{Abstract}
\end{abstract}


%

\allowdisplaybreaks		
\acresetall		
\maketitle

\section{Introduction}
\label{sec:introduction}
\input{02_Introduction}

\section{The model}
\label{sec:model}
\input{03_model}

\section{The trade-off between waiting time and matching costs}
\label{sec:lunch}
\input{04_tradeoff}

\section{Interpolating between waiting time and matching cost}
\label{sec:natural}
\input{05_paretoschedules}

\section{A balanced social planner}
\label{sec:planner}
\input{06_balancedschedule}


\section{Generalization}
\label{sec:generalizations}
\input{06b_generalizations}
\section{Discussion}
\label{sec:discussion}
\input{07_discussion}


\appendix

\section*{Appendix}

\section{Proof of Theorem \ref*{THM:COST}}
\label{app:match}
\input{10_appendix_match}

\section{Proof of Theorem \ref*{THM:WAIT}}
\label{app:wait}
\input{11_appendix_wait}

\section{Proof of approximation in Proof of \cref{prop:socialplanner}}
\label{app:tech}

\input{12_appendix_other}

%

\bibliographystyle{te_BP}	
\bibliography{../bibtex/Citations}

\end{document}

%% file: Abstract.tex
%
%

We examine two-sided markets where players arrive stochastically over time and are drawn from a continuum of types.
The cost of matching a client and provider varies, so a social planner is faced with two contending objectives:
\begin{inparaenum}
[\itshape a\upshape)]
\item
to reduce players' \emph{waiting time} before getting matched;
and
\item
to form efficient pairs in order to reduce \emph{matching costs}.
\end{inparaenum}
We show that such markets are characterized by a \emph{quick or cheap} dilemma: Under a large class of distributional assumptions, there is no `free lunch', \ie there exists no clearing schedule that is simultaneously optimal along both objectives. We further identify a unique breaking point signifying a stark reduction in matching cost contrasted by an increase in waiting time. 
Generalizing this model, we identify two regimes: one, where no free lunch exists; the other, where a window of opportunity opens to achieve a free lunch. Remarkably, greedy scheduling is never optimal in this setting.

%% file: 02_Introduction.tex

Many economic interactions require the dynamic matching of heterogeneous agents that arrive stochastically
to a two-sided market.
Examples include the dynamic matching
of clients and providers in markets for jobs and services,
of buyers and sellers in financial markets,
of taxis and passengers on road networks, 
of donors and recipients in organ exchanges, etc.\footnote{More generally, we focus on markets for  `nondurables';
for classic studies on product durability and market performance see \citet{Smi62}, \citet{Smi88}, and \citet{Dic12}.}

It is known that many of these markets vary substantially in terms of efficiency \citep{Rot94,Rot97}.
The focus of our investigation is on a crucial aspect of market design in this context, namely the \emph{scheduling of clearing events}. 
The goal of ``making a thick market'' \citep{Rot12} is to find the best schedule of market clearing so that sufficient clients and providers are in the market to allow for stable and efficient matches over time while not waiting excessively.
Designing an optimal clearing policy thus requires optimizing along the following two objectives:
\begin{enumerate}
\item
To reduce the coexistence of agents on the two sides of the market.
\item
To match parties in such a way so as to minimize cost (or maximize productivity).
\end{enumerate}
In pursuit of these two goals, clearing schedules need to be formulated to address the following key question:
\emph{How long should the social planner wait between two clearing events?}

To illustrate the above, consider the example of a governmental employment bureau faced with a dynamically evolving job market where job offers and job seekers arrive to the system stochastically over time. The bureau has two aims, namely to reduce the coexistence of vacancies and job seekers, and to match vacancies with the skills of individual job seekers so as to maximize productivity.
Waiting times incur costs via unemployment benefits, as well as costs due to productivity losses incurred by badly staffed vacancies.

To gain in generality, we abstract away from application-specific details
(such as the particular structure of the
application and recruitment processes). 
This allows us to focus on the trade-offs between two different and concurrent objectives, waiting time  and matching cost. Perhaps surprisingly, this \emph{quick or cheap} dilemma is not easily resolvable as greedy scheduling policies are generally not optimal in this context.


\Paragraph{Related work}
\label{sec:related}


Dating back to the 1950s, the first related strand of work focuses on behavioral aspects underlying the dynamics of unemployment and job vacancies in labor markets \citep{Dow58}.
These analyses identify avenues to reduce \emph{waiting} \textendash\ \ie the coexistence of unemployment and vacancies \textendash\ by better understanding the behavior of job seekers and job providers.
Lines of reasoning proposed to explain the coexistence of unemployment and vacancies include the classical search models of \citet{McC70}, \citet{Mor70}, and \citet{Luc74}, as well as more recent models with workforce inertia due to \citet{Shim07}.%
\footnote{Note that \citet{Shim07} terms his explanandum ``mismatch'' (as opposed to ``waiting''), a term the matching literature uses to describe suboptimal matchings, which may be confusing.}
We complement this literature with a view that some degree of waiting is actually beneficial from a social welfare perspective as it enables market thickening -- which in turn enables mismatch reduction.
To illustrate this, consider the example of \citet{Shim07},  where some laid-off steel workers are not immediately given vacant positions as nurses. This may indeed be deemed optimal by a social planner when \textendash\ by delaying their match \textendash\ these nurse vacancies eventually are taken up by better nurses and the jobless steel workers find other jobs in the steel industry that might become available in the future.%
\footnote{\emph{Waiting} is explained behaviorally through inertia in \citet{Shim07}, that is, by the argument that steel workers stay close to their factories hoping that they reopen;
\citet{Luc74} propose a different interpretation whereby waiting is due to the fact that steel workers must actively spend some time searching for these nursing jobs elsewhere.}

The second strand of related work comes from the matching literature
and extends the canonical static matching framework
to a dynamic setting.%
\footnote{The canonical static frameworks underlying our analyses were pioneered by \citet{Koe31,Ege31}, and \citet{Edm65};
see also \citet{Gal62} for matching with ordinal preferences.}
As in the example of steel workers and nurses above, \emph{mismatch} in dynamic environments may occur due to temporal inconsistencies, whereby, a posteriori, better matches were precluded by inferior matches that were formed earlier on.
Therefore, some delay may be optimal from a social planner perspective in order to reduce mismatch.
From a practical viewpoint, the challenge is to identify optimal mechanisms that thicken and clear the market in a way that balances these two objectives.

In this regard, \citet{Akb17}, \citet{Ash19}, \citet{Bac18}, and \citet{Loe16} break new ground in identifying optimal clearing schedules.%
\footnote{These are inspired by some earlier papers on dynamic matching in organ exchange by \citet{Zen02, Unv10}.
See \citet{Akb17} for a discussion.
See also \citet{Blo12}, \citet{Kur14}, and \citet{Les12} who study related queuing models where one side of the market is already present (such as in the housing market).} More precisely,
\citet{Akb17}, in the spirit of an organ exchange application such as the `kidney exchange', identify the optimal mechanism to maximize the number of matches, that is, to minimize the number of agents perishing resulting from failing to get recipients matched with donors in time.
In the model of \citet{Akb17}, agents from both sides of the market arrive and leave stochastically and all carry identical match values, \ie they are of the same type (in the spirit of each life being worth the same).
However, there are two types of agents, as some matches are feasible and others are infeasible, thus rendering some agents easier \textendash\ others harder \textendash\ to match.%
\footnote{This can be modeled by means of a dynamically changing compatibility graph where edges represent feasible matches.}
The optimal mechanism identified by \citet{Akb17} minimizes the number of unmatched patients based on information concerning arrivals and departures, which may involve delaying compatible matches.
Without such information, greedy scheduling is always optimal in this setting. In fact, \citet{Ash19} show that greedy policies are generally optimal, even if information about departure times is available when the above kind of `kidney exchange' markets becomes large.

In a related setting, \citet{Bac18} and \citet{Loe16} introduce binary diversifications of agents and the notion of waiting costs instead of perishing rates as in \citet{Akb17}.
In \citet{Bac18}, agents arrive in donor-recipient pairs and recipients are allowed to decline matches in order to remain in the market. This is motivated by the applications under scrutiny which include, among others, child adoption.
As a result, one of the study's key focuses is on strategic incentives and their role in determining market outcomes.
Their optimal clearing policy is discriminatory, in that it involves matching same-type pairs greedily, and delaying up to some threshold when there are only cross-type pairs in the market.
By contrast, \citet{Loe16} introduce a common discount factor (instead of a constant waiting cost) for both types of agents (as well as for the social planner).\footnote{The discount factor is motivated by the study's focus on financial markets, and would be determined by risk-free rate, beta, and risk premium. Related analyses of financial markets include \citet{Bud15} and \citet{Wah15} (see also \citet{Wah17} for market-making more generally) who consider periodic
clearing of order books to abate certain market phenomena such as volatility that stem from high-frequency trading.} They focus on the analysis of the possibility of efficient trade and rent extraction by the market maker.

In theoretical computer science, the study of related questions dates back at least to the pioneering paper of \citet{Kar90}.%
\footnote{\citet{Kar90} and subsequent work \textendash\ similar to its economic counterparts \textendash\ focus on models with two market sides, where by contrast one side is typically present to begin with and incoming agents from the other side can only match with some of the present agents according to a compatibility graph (see \citet{Meh13} for an overview and \citet{Agg11} for extensions to vertex-weighted matching).}
To the best of our knowledge, \citet{Eme16} were the first in this strand of research to consider the scenario where all agents arrive on the market over time (instead of just one market side).
They present a non-bipartite model
where requests arrive stochastically from one of $n$ different locations to study the performance of different algorithms in terms of worst-case matching and waiting cost.%
\footnote{\citet{Aza17} obtain additional results in terms of upper and lower bounds for the original model. \citet{Eme19} obtain sharper results for a two-location model.
There are also other extensions such as allowing for a stochastic graph \citep{And15,Ash18b}.} 
In the setting of \citet{Eme16}, the specific match costs result from the distance between agents' locations so that, for a patient social planner, it is optimal to wait and only match agents who are at the same location.

Finally, motivated by ride-sharing applications, \citet{Ash17} extend the model of \citet{Eme16} to a bipartite setting where agents independently arrive at different locations.
\citet{Ash17} study the performance of a family of \aclp{CS} along the two axes of waiting vs. mismatch separately, an approach that we extend in order to formulate the induced trade-off between waiting time and the cost of matching.%

\Paragraph{Contributions of the paper}
\label{sec:contribs}

Our paper examines dynamic markets with an infinite type space (in contrast to one or two types), a framework we call the \emph{dynamic clearing game}.
Our point of departure is the static assignment game of \citet{Sha72} to which we add a dynamic layer whereby clients and providers arrive to the market stochastically and independently.
Skills are drawn from a large class of type distributions so that every match is possible but some matches are more costly than others.
At each \acl{ME}, the social planner must decide who to match with whom, and how long to wait before the next \acl{ME}.
As such, the social planner is called to weigh, on the one hand,
mismatches incurred from matching clients and providers suboptimally; and, on the other hand, the
agents' waiting time.
To address this dual issue, we study \aclp{CS} in terms of \emph{when} to match in a fully heterogeneous setting where the social planner has no information regarding the cost of matching couples currently in the market and has no information about the future arrivals of individuals.


For concreteness, we start by studying a micro-level model where
costs of individual matches are distributed according to independent exponential random variables. Whilst our results hold for other distributions too, this model has the advantage of being tractable in closed form. 
In more detail, we first establish a class of optimal \aclp{CS} for two extreme types of single-objective social planners \textendash\ that is, for social planners who only care about minimizing waiting time (in which case greedy is best) or mismatch costs (resulting in endless delay), but not both at the same time.
Second, we show that these two objectives are mutually incompatible, and
multi-objective social planners (who care about both) face a fundamental trade-off.
Specifically, there is no `free lunch', that is \emph{there is no \acl{CS} that is approximately optimal in terms of both waiting time and matching cost.} Remarkably, the greedy \acl{CS} is sub-optimal for every multi-objective social planner.

Building on the no free lunch result, we proceed to fill the spectrum between matching cost and waiting time minimization.
We do so by introducing a class of \aclp{CS} covering a wide range of social planning desiderata between waiting time and matching cost,
and achieving a continuous trade-off between the two.
To explore the finer aspects of this trade-off, we introduce a utility model for the social planner whereby the associated utility of matching cost is of the same order as the agents' utility of waiting time.
Under this model, we show that there exists a non-trivial \acl{CS} achieving this balance, and we show that this schedule is effectively unique (up to asymptotic order considerations).

Finally, we generalize our key findings by studying different decay rates of matching costs (instead of focusing on one decay rate that results from the micro-founded match costs).
We identify two regimes. One, where no free lunch continues to hold. The other, where the benefit from waiting is growing quickly enough, such that a window of opportunity opens and it \emph{is} possible to get a free lunch.
As before, in both regimes, greedy scheduling is generally sub-optimal.

Compared to the existing literature on the trade-off between waiting and mismatch (both in economics and computer science) our model introduces \emph{incomplete information} about the distribution of past and future match costs and considers an infinite type space in a tractable model.
As a consequence, the social planner tries to resolve the trade-off between matching optimally and waiting time in light of incomplete information.
Incomplete information in our setting implies that the social planner must employ \aclp{CS} that do not take as input the relative strengths of current and future matches (since the latter is unknown), thus yielding qualitatively new results.

In contrast to prior results for markets with one or two types of match costs where lack of information resulted in optimality of some form of greedy scheduling, we find that greedy clearing is generally not optimal in the presence of many types. 
Hence, the quick-versus-cheap trade-off is more intricate than previously found.
Moreover, our results may actually also have consequences for applications that have been studied before too (e.g. kidney exchange) if other match value metrics (e.g. potential years of life lost or disability-adjusted life years) are used that would produce more than binary match values.
By studying fully heterogeneous match costs we have to rely on different mathematical tools compared to previous analyses, which were often able to reduce the induced dynamics to discrete Markov processes.

The key technical innovations of our paper concern the concurrent consideration of a continuum of types, independent arrivals, and incomplete information. In turn, these contributions rely on a range of previously unused tools from probability theory and disordered systems to obtain closed-form solutions.
These underlying results are concerned with the expected matching cost for given instances of random, static assignment games.
In particular, in static assignment games with the same number of clients and providers and $\exp(1)$ distributed edge weights, \citet{Ald01} proved the long-standing conjecture that the expected minimum weight matching converges to $\pi^2/6$ (\ie as the number of players is growing). This result was later extended by \citet{Wae05} to assignment games with match costs drawn from non-identical exponential distributions.%
\footnote{To the best of our knowledge, the work of \citet{Wal79} is the first to pose the question, while \citet{Mez87} conjectured the specific limit value.
We also leverage the analyses of \citet{Buc02} and \citet{Lin04} who obtain results for the expected values of finite instances of the latter models, showing \textendash\ as a byproduct \textendash\ that the value is increasing with the number of agents.
For a survey of this literature, we refer the reader to \citet{Kro09}.}
By leveraging the techniques of \citet{Ald01} and \citet{Wae05}, we are able to compute the expected matching cost for every `snapshot in time' of the dynamic clearing game.
This provides strong foundations for our proofs which are then focused on estimating the fluctuations that result from the random arrival of clients and providers and their randomly drawn match costs.
To achieve this, we use several approximation techniques (in particular, the approximation of the arrival process by a continuous-time Wiener process),
which allow us to port over several results from martingale limit theory (such as the law of the iterated logarithm).

\Paragraph{Paper outline.}
The rest of the paper is structured as follows.
In \cref{sec:model}, we introduce the \emph{dynamic clearing game} and the performance measures relevant for our analysis.
In \cref{sec:lunch}, we show that a social planner who cares about both waiting and mismatch faces a fundamental and non-negligible trade-off.
We then go on to analyze a natural selection of \aclp{CS} in \cref{sec:natural}, which cover the whole range of possible trade-offs.
In \cref{sec:planner}, we commit to a specific utility function that specifies how the social planner values waiting time versus matching cost and find the unique optimal \acl{CS}.
\cref{sec:generalizations} generalizes the analysis and shows that there are two regimes, one where `free lunch' is not achievable and one where it is achievable.
Finally, in 
\cref{sec:discussion}, we discuss practical implications and of avenues for future research.

%% file: 03_model.tex

In this section, we introduce the model, which we shall refer to as the \emph{dynamic clearing game}.

\Paragraph{The dynamic clearing game.}
Consider the following model of a dynamic two-sided market evolving in continuous time $\tau\in[0,\infty)$.
At each tick of a Poisson clock with rate $1$ an agent enters the market;
this agent could be either a \emph{client} or a \emph{provider}, with equal probability.%
\footnote{We are using here the terms `client' and `provider' in a generic sense, just to illustrate the difference between the two sides of the market.
As we explain below, what is important from a modeling perspective is that `clients' are to be matched to `providers' (as in our running example of job vacancies and job seekers).}
To keep track of the number of agents in both sides of the market, let $\boys(\tau)$ and $\girls(\tau)$ denote the set of clients and providers that have entered the market by time $\tau$ (and possibly already left again), and let $\nBoys(\tau) = \abs{\boys(\tau)}$ and $\nGirls(\tau) = \abs{\girls(\tau)}$ be the respective numbers thereof.
Then, the number of agents on the \emph{short side of the market} will be written $\nShort(\tau) = \min\{\nBoys(\tau),\nGirls(\tau)\}$.%
\footnote{In a slight (but convenient) abuse of notation, we will sometimes write $\nBoys(t)$, $\nGirls(t)$, and $\nShort(t)$ to denote respectively the number of clients, providers, and agents at the short side of the market when the $t$-th agent enters the market \textendash\ specifically, letting $\tau(t)$ denote the time at which the $t$-th agent enters the market, we will write $\nBoys(t) \equiv \nBoys(\tau(t))$, etc.}

As in the static assignment model of \citet{Sha72} on which we build, we consider a one-to-one matching market where each client is to be \emph{matched} to at most one provider and vice versa;
then, once a couple is matched, both agents leave the market.
For example, in labor market language, each job seeker gets at most one job, and each vacancy concerns exactly one worker;
once a match has been made, the governmental job bureau removes the matched pair from its ledger, and the process continues.

For concreteness, we shall next define a specific family of match cost distributions. 
This is done solely to streamline our presentation: Our methodology allows us to be more general as we discuss in Section \ref{sec:generalizations}.
Suppose that the quality of a (candidate) pair is characterized by an inherent
\emph{match parameter} $\rate_{\boy\girl}$ where a higher parameter will represent a lower expected match cost. Match costs are independently and exponentially distributed with rate $\rate_{\boy\girl}$.\footnote{As mentioned by \citet{Ald01} and developed in detail by \citet[Section 2]{Jan99} generalizations to larger classes of distributions are easily obtained. For ease of exposition we stick to exponential distributions with the exception of Section \ref{sec:generalizations} that generalizes our main results.}

Specifically, we posit that the \emph{match cost} $\match_{\boy\girl} > 0$ when client $\boy\in\boys$ is matched to provider $\girl\in\girls$ is an independent draw from
 an exponential distribution of rate $\rate_{\boy\girl}$ for any time $\tau$, that is,  $\match_{\boy\girl} \sim \exp(\rate_{\boy\girl})$.
For example, a popular model assumes that $\rate_{\boy\girl}$ is composed by additively separable components describing the agents' types and a couple-specific term depending possibly on both the identity of the agents and their types \citep{Kan18}.
For generality, our only assumption regarding the rate parameters $\rate_{\boy\girl}$ is that they are bounded from below by $\underline\rate$, from above by $\overline\rate$, and have mean value $\rate$, that is, $\rate = \lim_{\tau\to\infty} \bracks{\nBoys(\tau) + \nGirls(\tau)}^{-1} \sum_{\boy=1}^{\nBoys(\tau)} \sum_{\girl=1}^{\nGirls(\tau)} \rate_{\boy\girl}$.

\Paragraph{The social planner.}
Throughout the sequel, we  assume the existence of a social planner who, whenever an agent arrives on the market, observes the arrival;
other than that, the social planner has no other information regarding the arrival process of the agents (or the distribution of their match costs).
Due to this lack of information, the social planner has no basis to judge whether a particular agent arriving in the market is `good' or `bad', and is thus left with the challenge of choosing a \acl{CS} with which to operate the market.
In the sequel, we will also write $A\equiv A(\tau)$ for the number of clients/providers that have been assigned a partner up to time $\tau$, and $R(\tau) = \nBoys(\tau) + \nGirls(\tau) - 2A(\tau)$ for the number of unmatched agents up to time $\tau$.

With all this in hand, a \acdef{CS} will be a rule that determines:
\begin{enumerate}[(i)]
\addtolength{\itemsep}{.5ex}
\item
At which points in time $\tau\in(0,\infty)$ to trigger a \acdef{ME}, possibly depending on $\nBoys(\tau)$, $\nGirls(\tau)$ and $A(\tau)$.
\item
Which players to match at a given \acl{ME}, possibly depending on the current match costs of agents who have already arrived to the market until time $\tau$.
\end{enumerate}
After a \acl{ME}, the players who are being matched leave the market, while the unmatched players remain on the market.

In what follows, we shall focus on \aclp{CS} that match a single couple per matching event. In particular, the \aclp{CS} we analyze will match the couple with the minimal matching cost in each matching event.\footnote{The only \acl{CS} that we consider and which violates this principle is the \acf{FCFS} \acl{CS} which we describe in \cref{sec:lunch}.} Restricting ourselves to these kinds of \aclp{CS} is motivated by our aim to study \aclp{CS} that can be paired with a market mechanism (\eg a two-sided auction). 
\cite{Mye83} and \cite{Rus94} show the general impossibility to have ex-post efficient and budget balanced mechanisms for two-sided market games with private information. Their results  rely on the assumption that, with positive probability, any given client-provider pair have valuations for each other such that trade is not individually rational for both at
any price. Since the latter doesn't hold in our setting, one could micro-found the interaction avoiding the impossibility, that is, define a mechanism that is ex-post efficient and budget balanced.
We shall assume throughout the analysis the existence of such a mechanism;
however, given our focus on the social planner an explicit analysis of such mechanisms is beyond the scope of the present paper.

Now, as discussed before, the social planner aims to match clients and providers optimally along two axes:
\begin{inparaenum}
[\itshape a\upshape)]
\item
to reduce the coexistence of clients and providers (\ie \emph{waiting time});
and
\item
to match clients and providers in a way that minimizes matching cost (\ie \emph{mismatch}).
\end{inparaenum}
Beginning with the latter, 
the \emph{expected matching cost} for the first $A$ couples is defined as
\begin{equation}
\cost_\CS(A)
	\equiv \exof*{\sum_{k=1}^A \match_{\boy_{k},\girl_{k}}}
\end{equation}
where $\match_{\boy_{k},\girl_{k}}$ is the match cost of the $k$-th matched couple and the expectation is taken with respect to the random arrival of clients and providers and the randomness of the match costs.
Similarly, the \emph{expected waiting time} of a \acl{CS} until time $\horizon$ is defined as
\begin{equation}
\wait_\CS(\horizon)
	\equiv \exof*{\int_0^{\horizon} R(\tau) \dd\tau}
\end{equation}
where the expectation is taken with respect to the random arrival of agents in the market.
In the sections that follow, we will explore the optimization of these two performance metrics, and the trade-offs that arise when trying to minimize both.

%% file: 04_tradeoff.tex

Our analysis begins with the case of a single-minded social planner.
Specifically, we investigate which \acl{CS} a social planner would employ if either only caring about the expected waiting time, or only caring about the expected matching cost.
After we deal with these two cases separately, we shall proceed to show that these objectives are mutually incompatible and lead to an unavoidable trade-off for the social planner.

\Paragraph{Single-minded social planners.}
First, a social planner who is optimizing the agents' expected waiting time will choose a \acl{CS} which leaves no unmatched couples at any point in time.
To do so, we will consider a `greedy' \acl{CS}, denoted $\CS_{\greedy}$, which performs a minimum weight matching whenever there is exactly one unmatched client/provider pair in the market.
Second, a social planner who is optimizing the agents' expected matching cost will choose a \acl{CS} which \textendash\ ideally \textendash\ waits until everyone has arrived in the market and then matches agents optimally (thus minimizing the sum of match costs).%
\footnote{To make such a \acl{CS} realistic, all agents would need to arrive in the market in finite time;
since this schedule will mostly serve as a theoretical comparison baseline, we will not consider this issue in detail.}
That is, the hypothetical `patient' \acl{CS}, denoted $\CS_{\patient}$, should be preferred by any social planner who is only concerned with the expected matching cost.

The implementation of these schedules leads to the following matching cost and waiting time:

\begin{proposition}
\label{prop:optimalbaseline}
The optimal \aclp{CS} for a single-objective social planner are:
\begin{enumerate}
[\noindent \upshape(1)]
\item
The patient \acl{CS} $\CS_{\patient}$
is optimal with respect to matching cost minimization;
in particular, for all $A\geq 1$, we have:
\[
\frac{\log 2}{\overline\rate}
	\leq \cost_{\patient}(A)
	\leq \frac{\pi^{2}}{6 \underline\rate}
\]
\item
The greedy \acl{CS} $\CS_{\greedy}$
is optimal with respect to waiting time minimization;
in particular, for all $\tau\geq 0$, we have:
\[
\wait_{\greedy}(\tau)
	= \frac{2}{3} \tau^{3/2}
\]
\end{enumerate}
\end{proposition}

\begin{remark*}
In view of \cref{prop:optimalbaseline}, the expected matching cost of $\CS_{\patient}$ and the expected waiting time of $\CS_{\greedy}$ will serve as the benchmark for comparing the matching cost and waiting time of any other \acl{CS}.
\end{remark*}

\begin{proof}[Proof of \cref{prop:optimalbaseline}.]
We prove our claims for each of the two \aclp{CS} separately.

\begin{proofpart}{Matching cost minimization}
For the first assertion, note that the exponential distribution is closed under scaling by a positive factor, \ie if $X\sim \exp(\kappa)$ then $\mu X\sim \exp(\kappa/\mu)$.
In our case, this implies that
\begin{equation}
\match_{\boy\girl}
	\sim \exp(\rate_{\boy\girl})
	\iff
\match_{\boy\girl}
	\sim \frac{1}{\rate_{\boy\girl}}\exp(1)
\end{equation}
We have that for all $\boy\in\boys$, $\girl\in\girls$, the distribution of $\match_{\boy\girl}$ is first-order stochastically dominated by $\underline\rate^{-1} \exp(1)$.
Thus the expected weight is upper bounded by the simplified problem where all match costs are distributed according to $\underline\rate^{-1} \exp(1)$.
With this in mind, we will simplify notation in the rest of the proof by setting $\underline\rate = 1$.

By the summation formula of \citet{Buc02} and \citet{Lin04}, we have for the expected weight of the minimum $A$-matching (note that $A=\nShort$, recalling that $\nShort = \min\{\nBoys,\nGirls\}$):
\begin{equation}
\mathbb E_{\min}\left [\sum_{k=1}^{\nShort} \match_{\boy_{k},\girl_{k}}\right ]
	= \sum_{\substack{\boy,\girl\geq 0\\ \boy + \girl <\nShort}} \frac{1}{(\nBoys - \boy) \cdot (\nGirls - \girl)}.
\end{equation}
Thus, we readily have
\begin{equation}
\cost_{\patient}(\nShort)
	= \mathbb E_{\min} \left[\sum_{k=1}^{\nShort} \match_{\boy_{k},\girl_{k}}\right]
	\leq \sum_{\substack{\boy,\girl\geq 0\\ \boy + \girl <\nShort}} \frac{1}{(\nShort-\boy)\cdot(\nShort-\girl)}.
\end{equation}
To proceed, by \citet[Lemma 3.1]{Wae09} we have
\begin{equation}
\sum_{\substack{\boy,\girl\geq 0\\ \boy + \girl <\nShort}} \frac{1}{(\nShort-\boy)\cdot(\nShort-\girl)}
 	= \sum_{k=1}^{\nShort} \frac{1}{k^{2}}
	\leq \zeta(2),
\end{equation}
where $\zeta(2) = \sum_{n=1}^{\infty} 1/n^{2} = \pi^{2}/6$ is the Basel constant.
Returning to our original problem, we conclude that $\cost_{\patient}(A)\leq \pi^{2}/(6\underline\rate)$, as claimed.

To compute the lower bound, we need to consider the expected match cost for $A=1$, because matching more players would only serve to increase the expected matching cost.
In the patient \acl{CS} for $A=1$, the process terminates when at least one client and at least one provider have entered the market.
Let $Y\geq 2$ be the number of agents required to observe at least one client and one provider.
Then the event $Y=k+1$ is the same event as the union of the disjoint events `the first $k$ agents are clients and the $(k+1)$-th agent is a provider' and `the first $k$ agents are providers and the $(k+1)$-th agent is a client'.
Each of the latter events has probability $1/2^{k+1}$, so $\probof{Y=k+1} = 2^{-k}$.
Moreover, as we prove in \cref{lemma_minofexpvariables},
for $Y=k+1$, the expected minimum match cost is given by $\frac{1}{\overline\rate\cdot k}$.
Thus for $A=1$, we get
\begin{equation}
\cost_{\patient}(N)=\sum_{k=1}^\infty \frac{1}{\overline\rate k} \probof{Y=k+1}
	= \frac{1}{\overline\rate} \sum_{k=1}^\infty \frac{1}{2^{k} k}
	=\frac{\log2}{\overline\rate},
\end{equation}
where the last equality follows from the series expansion $\log(1-x) = -x - x^{2}/2 - x^{3}/3 - \dotsm$ applied to $x=1/2$.
\end{proofpart}

\begin{proofpart}{Waiting time minimization}
For our second assertion, note that, at any point in time, there are either no clients or no providers in the market.
In view of this, let $S_\tau$ be the difference of clients and providers who have arrived to the market until time $\tau$, that is, $S_\tau=\nBoys(\tau) - \nGirls(\tau)$.
Then, for all $\horizon > 0$, we get:
\begin{equation}
\wait_{\greedy}(\horizon)=\exof*{\int_0^{\horizon} \abs{S_\tau} \dd\tau}
	= \int_0^{\horizon} \exof{\abs{S_\tau}} \dd\tau
\end{equation}
where the latter equality holds by Tonelli's theorem (since $\abs{S_\tau}$ is non-negative).

Applying Tonelli's theorem a second time, we can consider the case where the expectation with respect to the arrival times is taken first.
To do so, consider the process where at the fixed points in time $\tau=1,2,\dotsc$ an agent arrives to the market and let $\bar S_\tau$ be the difference of clients and providers who have arrived to the market at time $\tau$.
We then have:
\begin{equation}
\label{WienerWiener}
\exof*{\int_0^{\horizon} \abs{S_\tau} \dd\tau}
	= \int_0^{\horizon} \exof{\abs{S_\tau}} \dd\tau
	= \int_0^{\horizon} \exof{\abs{\bar S_\tau}} \dd\tau
\end{equation}
It is well-known that for $\tau\to\infty$ the appropriately rescaled random walk $\bar S_\tau$ converges in distribution to the Wiener process $W_\tau$ \citep{Kac47}.
Thus, for large $\horizon$, \cref{WienerWiener} gives
\begin{equation*}
\exof*{\int_0^{\horizon} |S_\tau|d\tau}
	= \int_0^{\horizon} \exof{\abs{W_\tau}} \dd\tau
	= \int_0^{\horizon} \sqrt{\Var(W_\tau)} \dd\tau
	= \int_0^{\horizon} \sqrt{\tau} \dd\tau
	= \frac{2}{3}\horizon^{3/2}.
	\qedhere
\end{equation*}
\end{proofpart}
\end{proof}	
This concludes the analysis of a single-minded social planner who either only cares about the expected waiting time, or onlythe expected matching cost.

\Paragraph{Multi-objective social planners.}
Going beyond the narrow view of a single-minded social planner, we proceed below to examine the case of social planners that care about \emph{both} the expected cost of matching and the agents' overall expected waiting time.
Natural candidates to evaluate a \acl{CS} in this context are the expected matching ratio and the expected waiting ratio,
defined below as follows:

\begin{enumerate}
\item
The \emph{expected matching ratio} of a \acl{CS} $\CS$ is
\begin{flalign} \label{eq:matchingratio}
\alpha
	&\equiv \alpha(A)
	= \frac{\cost_{\CS}(A)}{\cost_{\patient}(A)}.
\intertext{\item  
The \emph{expected waiting ratio} of a \acl{CS} $\CS$ is} \label{eq:waitingratio}
\beta
	&\equiv \beta(\tau)
	= \frac{\wait_{\CS}(\tau)}{\wait_{\greedy}(\tau)}.
\end{flalign}
\end{enumerate}
Note that  $\CS_{\patient}$ is optimal with respect to expected matching cost while $\CS_{\greedy}$ is optimal with respect to expected waiting time.

Going forward, note that all candidate \aclp{CS} can be characterized by a function $f\from \mathbb{R}^{+} \to \mathbb{R}^{+}$ such that the $k$-th ($k\in \mathbb N$) couple is matched
when $\ceil{f(k)}$ agents are on the short side of the market. Denote a \acl{CS} that is defined via a function $f$ by $\CS_f$. 
Without any restrictions on $f$ this includes all possible \aclp{CS} that always match the couple with the minimal match cost. 
However, given that we wish to analyze the asymptotic regime where enough agents have entered the market, we focus below on a natural class of functions introduced by \citet{Har10old} which make such comparisons possible. Specifically, each function in this class  is  defined, for all $x\geq0$, by a finite combination of the basic arithmetic operations (addition, multiplication, raising to a power, and their inverses), operating on the variable $x$ and on real constants.
 \citet[Theorem, page 18]{Har10old} shows that for any two such functions, $f$ and $g$, either $f=\omega (g), f=\Theta (g)$, or $f=o(g)$.

Then, for such functions, we will use the following asymptotic notations:\\ $f(x) = \bigoh(g(x))$ if $f(x)<c\cdot g(x)$ for some $c>0$ constant and $x$ sufficiently large.  $f(x)=\Omega(g(x))$ is the inverse $O$ notation ($f(x)>c\cdot g(x)$ for $x$ sufficiently large).    $f(x)=\Theta(g(x))$ if  there exist two constants $k,K\geq 0$ and a positive integer $x_0$ such that $kg(x)\leq f(x)\leq Kg(x)$ for all $x\geq x_0$. \\ For $g(x)$ non-zero $f(x)=o(g(x))$ if $\lim_{x\to\infty}\frac{f(x)}{g(x)}=0$ and $f(x)=\omega(g(x))$ if $\lim_{x\to\infty}\frac{f(x)}{g(x)}=\infty$.


In light of the above, the first question that arises is whether there exists a \acl{CS} that is optimal along \emph{both} axes (at least, asymptotically).
To formalize this, we say that a \acl{CS} $\CS$ has \emph{finite expected matching ratio} if $\limsup_{A(\tau)>0} \alpha(A(\tau)) < \infty$;
likewise, we say that \acl{CS} has \emph{finite expected waiting ratio} if $\limsup_{\tau>0} \beta(\tau) < \infty$.

The following theorem shows that the answer to the above question is a resounding `no':
\begin{theorem}[`No free lunch']
\label{thm:freelunch}
There exists no \acl{CS} simultaneously achieving finite ratios for both expected matching cost and waiting time.
\end{theorem}

\cref{thm:freelunch} illustrates that multi-objective social planners are faced with a crucial trade-off independently of their specific utility function \textendash\ provided of course that they care about both matching cost and waiting time in a non-trivial way.
In addition, \cref{thm:freelunch} justifies the performance measures for the two dimensions of mismatch (expected matching ratio and waiting ratio) and in particular the sufficiency to analyze them in terms of orders of $\tau$ (or $A=A(\tau)$).

\begin{proof}[Proof of \cref{thm:freelunch}.] We prove this result by contradiction;
specifically, we find a necessary condition for a \acl{CS} to have finite expected matching ratio and then show that \aclp{CS} satisfying this condition cannot have a finite expected waiting ratio.

To make this precise, consider the \acl{CS} $\CS_{f}$ that matches the $k$-th couple when $\nShort - (k-1) = \ceil{f(k)}$.  
We shall show that a necessary condition for a \acl{CS} $\CS_{f}$ to have finite expected matching ratio is 
\begin{equation}
\label{eqlimitsmall}
f(k)
	= \omega(k^{1/2})
\end{equation}
To show this, consider the \acl{CS} that matches the $k$-th couple when at least $\ceil{k^{1/2}}$ players are on the short side of the market;
with a fair degree of hindsight, denote this \acl{CS} as $\CS_{\gamma=1/2}$.%
\footnote{See \cref{sec:natural} for detailed definitions.}
As we show in \cref{THM:COST}, this \acl{CS} has $\alpha(A) = \Theta (\log A)$.
Thus, in order for another schedule $\CS_{f}$ to have finite expected matching ratio,
\aclp{ME} have to happen orders of magnitude later than in $\CS_{\gamma=1/2}$.
Concretely, for $k$ large enough the $k$-th couple is cleared at time $\tau_{\CS_{f}}(k)=\omega (\tau_{\CS_{\gamma=1/2}}(k))$.
It follows that \cref{eqlimitsmall} holds.
 
We can now analyze the expected waiting time for $\CS_{f}$ such that \cref{eqlimitsmall} holds for $f$. 
To construct a lower bound, consider the alternative arrival process, where clients and providers alternatingly arrive to the market. Note that for any given \acl{CS} this process incurs lower waiting time. For the \acl{CS} we consider the waiting time of this alternative arrival process is precisely governed by the fact that the $k$-th match takes place when at least $f(k)$ players are on the short side of the market.
Further, note that $\tau (A)$ is clearly smaller for this new arrival process compared to the original process. Given that we only need to show that the waiting time is increasing in $A$ it suffices to show that it is increasing in $\tau$ (not conditioned on $A$). Thus, the waiting time is lower bounded by using the approximation by the Wiener process (as in the proof of \cref{prop:optimalbaseline}) and by observing that arrival is governed by a Poisson clock of rate 1:
\begin{equation}
\int_0^{\horizon} 2f(\tau) d\tau
	= \omega\parens*{ \int_0^{\horizon} 2\sqrt \tau \dd\tau }
	= \omega (\horizon^{3/2})
\end{equation}
By \cref{prop:optimalbaseline}(ii), the optimal expected waiting time is $(2/3)\,\horizon^{3/2}$, so we conclude that the expected waiting ratio is lower bounded by $\omega(\horizon^{3/2}) / \horizon^{3/2} = \omega(1)$ and our proof is complete.
\end{proof}

This concludes our first result for multi-objective social planners, showing that the trade-off between cost of matching and waiting time is essential.

%% file: 05_paretoschedules.tex

In this section, we analyze a class of \aclp{CS} covering a broad spectrum of social planning desiderata interpolating between matching cost and waiting time. 

To begin, recall that \cref{prop:optimalbaseline} provides the expected matching cost of the patient \acl{CS} $\CS_{\patient}$ (which minimizes mismatches) and the expected waiting time of the greedy \acl{CS} $\CS_{\greedy}$ (which minimizes waiting times).
Interpolating between these two `extreme' schedules, we shall consider below a class of \aclp{CS} where the social planner waits for some length of time in order to accrue some intermediate number of agents on both sides of the market.
Concretely, we shall study \aclp{CS} that match the $k$-th couple when $\nShort - k = f(k)$, \ie when $f(k)$ players are on the short side of the market.\footnote{Recall that $\nShort = \min\{\nBoys,\nGirls\}$.}

For concreteness, we restrict ourselves to \aclp{CS} of the form 
\begin{flalign}
f(k)
	= \Theta (k^\gamma)
	\quad
	\text{for some $\gamma\in[0,1]$}.
\end{flalign}
\label{asfdasdfasdf}
For $\gamma = 0$, the induced clearing schedules match players once a constant threshold is reached;
in particular, the greedy schedule is recovered when $f\equiv 1$ (corresponding to $\gamma=0$).
More generally, we shall denote clearing schedules of the above form by $\CS_\gamma$ and write $\CS_{\gamma=1/2}$ for the clearing schedule with $\gamma=1/2$.
Similarly we shall use the notation
$\alpha_\gamma$ for the expected matching ratio of $\CS_\gamma$ and $\beta_\gamma$ for the expected waiting ratio of $\CS_\gamma$.

In addition to the clearing schedules induced by the assumptions above, we shall also consider another natural schedule based on the principle of \acdef{FCFS}, \ie when agents are matched as soon as possible on a \acl{FCFS} basis. 
This schedule, which we denote by $\CS_{\FIFO}$, differs from $\CS_{\greedy}$ in terms of who is matched with whom (\acl{FCFS} vs. minimum cost matching) but not regarding when a matching event occurs. 
As such, given that $\CS_{\FIFO}$ does not take into account matching costs, it is not reasonable to expect that it will perform well on any dimension other than the agents' expected waiting times.
On the other hand, it exhibits `fairness' relative to the agents' arrival times, a feature which is crucial in many applications.%
\footnote{Indeed, this may be a desirable feature in applications such as processor time requests in distributed computing. We shall leave extensions of our analyses to include fairness considerations for future work.}

\Paragraph{Overview of results.}
\cref{tab:schedules} summarizes all \aclp{CS} analyzed below (including a `balanced' schedule, $\CS_{\balanced}$, that we discuss in \cref{sec:planner}).
Our results (in terms of each schedule's expected matching and waiting ratio) are then summarized in \cref{tab:results}:
as can be seen, the family of schedules under study captures the full range between schedules that are `good' relative to mismatches and `bad' relative to waiting times, and vice versa.
\medskip


\begin{table}[t]
\centering
\renewcommand{\arraystretch}{1.5}
\small
\input{Tables/Schedules.tex}
\vspace{\medskipamount}
\caption{Overview of the various \aclp{CS} considered in the sequel. Note that all schedules other than $\CS_{\FIFO}$ match the couple with the minimum match cost at each \acl{ME}.}
\label{tab:schedules}
\end{table}




\begin{table}[h]
\renewcommand{\arraystretch}{1.5}
\small
\input{Tables/Results.tex}
\vspace{\medskipamount}
\caption{The range of expected matching and waiting ratios;
$\CS_{\balanced}$ is discussed in \cref{sec:planner}. Recall from \cref{eq:matchingratio} and \cref{eq:waitingratio}  the expected matching ratio
$\alpha
	\equiv \alpha(A)
	= \frac{\cost_{\CS}(A)}{\cost_{\patient}(A)}$ ($\cost_{\patient}(A)=\Theta(1)$)
and the expected waiting ratio
$\beta
	\equiv \beta(\tau)
	= \frac{\wait_{\CS}(\tau)}{\wait_{\greedy}(\tau)}$ ($\wait_{\greedy}(\tau)=\Theta(\tau^{3/2})$).
	}
\label{tab:results}
\end{table}


In view of these results, the \acl{CS} $\CS_{\gamma=1/2}$ can be seen as a \emph{phase transition} between two markedly different regimes.
On the one hand, for $\gamma < 1/2$, the expected matching ratio $\alpha(A)$ grows as a power law in $A$ while the expected waiting ratio $\beta(\tau)$ is finite.
On the other hand, for $\gamma > 1/2$, we have a finite expected matching ratio but an expected waiting ratio that grows polynomially.
Finally, at the critical point $\gamma = 1/2$, the expected matching ratio grows to infinity for large $A$, but at a slow, logarithmic rate ($\Theta(\log A)$).
Notably, the phase transition at $\gamma=1/2$ signifies a discontinuity of the expected matching ratio, so it is a \emph{first-order} phase transition;
by contrast, the expected waiting ratio exhibits no such discontinuity, signifying a \emph{second-order} phase transition.

The infinite matching ratio vis-a-vis the finite waiting ratio for $\gamma=1/2$ suggests that further fine-tuning should be possible and, indeed,  the `balanced' schedule $\CS_{\balanced}$ (which we define and discuss in \cref{sec:planner})
reduces the growth of the expected matching ratio by a factor of $(\log A)^{2/3}$
while increasing the  expected waiting ratio $\beta(\tau)$ by a factor of $(\log \tau)^{1/3}$.
In a sense (that we shall make precise in the next section) this is as close as we can get to a `free lunch' in this setting.

\Paragraph{Formal statements.}
We now proceed to provide complete statements of the results discussed above.
To streamline our presentation, we have relegated the detailed proofs to \cref{app:match,app:wait};
however, the main pattern of the proofs can also be seen in \cref{sec:planner} where we treat the case of $\CS_{\balanced}$.

We begin with our results for the matching cost ratio $\alpha$:

\begin{theorem}
\label{THM:COST}
The expected matching ratios for the schedules under study are as follows:
\begin{subequations}
\label{eq:cost-bounds}
\begin{alignat}{3}
&
		\text{\textpar{$\CS_{\FIFO}$}}&&\text{\; \ac{FCFS} matching:}
	&&\quad
	\alpha_{\FIFO}
		= \frac{6\underline \rate\rate}{\pi^{2}} A
\label{eq:cost-FIFO}
\\[3pt]
&
	\text{\textpar{$\CS_{\greedy}$}}&&\text{\; Greedy matching:}
	&&\quad
	\alpha_{\greedy}
		\geq \frac{6\underline\rate}{5\pi^{2}} A^{1/2}
\label{eq:cost-greedy}
\\[3pt]
&
	\text{\textpar{$\CS_{0\leq\gamma<1/2}$}}&&\text{\; Subcritical rate matching:}
	&&\quad
	\underline C_{\gamma} \, A^{1/2-\gamma}
		\leq \alpha_{0\leq\gamma<1/2}
		\leq \overline C_{\gamma} \, A^{1-2\gamma}
\label{eq:cost-subcritical}
\\[3pt]
&
	\text{\textpar{$\CS_{\gamma=1/2}$}}&&\text{\; Critical rate matching:}
	&&\quad
	\frac{2\underline\rate}{\pi^{2}} \log A
		\leq \alpha_{\gamma=1/2}
		\leq \frac{\overline\rate}{\underline\rate} \frac{1 + \log A}{\log 2}
\label{eq:cost-critical}
\\[3pt]
&
	\text{\textpar{$\CS_{1/2<\gamma\leq1}$}}&&\text{\; Supercritical rate matching:}
	&&\quad
	\alpha_{1/2<\gamma\leq1}
		= \frac{\overline \rate}{\underline\rate} \frac{\zeta(2\gamma)}{\log 2}
\label{eq:cost-supercritical}
\\[3pt]
&
	\text{\textpar{$\CS_{\patient}$}}&&\text{\; Patient matching:}
	&&\quad
	\alpha_{\patient}
		= 1
\label{eq:cost-patient}
\end{alignat}
\end{subequations}
\begin{remark*}
In the above, $\underline C_{\gamma}$ and $\overline C_{\gamma}$ are positive constants, and $\zeta(s) = \sum_{n=1}^{\infty} 1/n^{s}$ denotes the Riemann zeta function (so $\zeta(s)<\infty$ for all $s>1$). 
\end{remark*}
\end{theorem}

By contrast, for the expected waiting ratio $\beta$, we have:

\begin{theorem}
\label{THM:WAIT}
\begin{subequations}
\label{eq:wait-bounds}
The expected waiting ratios for the schedules under study are as follows:
\begin{alignat}{3}
&
		\text{\textpar{$\CS_{\FIFO}$}}&&\text{\;\; \ac{FCFS} matching:}
	&&\qquad
	\beta_{\FIFO}
		= 1
\label{eq:wait-FIFO}
\\[3pt]
&
	\text{\textpar{$\CS_{\greedy}$}}&&\text{\;\; Greedy matching:}
	&&\qquad
	\beta_{\greedy}
		= 1
\label{eq:wait-greedy}
\\[3pt]
&
	\text{\textpar{$\CS_{0\leq\gamma<1/2}$}}&&\text{\;\; Subcritical rate matching:}
	&&\qquad
	\beta_{0\leq\gamma<1/2}
		= \Theta(1)
\label{eq:wait-subcritical}
\\[3pt]
&
	\text{\textpar{$\CS_{\gamma=1/2}$}}&&\text{\;\; Critical rate matching:}
	&&\qquad
	\beta_{\gamma=1/2}
		= \Theta(1)
\label{eq:wait-critical}
\\[3pt]
&
	\text{\textpar{$\CS_{1/2<\gamma\leq1}$}}&&\text{\;\; Supercritical rate matching:}
	&&\qquad
	\beta_{1/2<\gamma\leq1}
		= \Theta(\tau^{\gamma-1/2})
\label{eq:wait-supercritical}
\\[3pt]
&
	\text{\textpar{$\CS_{\patient}$}}&&\text{\;\; Patient matching:}
	&&\qquad
	\beta_{\patient}
		= \Theta(\tau^{1/2})
\label{eq:wait-patient}
\end{alignat}
\end{subequations}
\end{theorem}
In closing this section, it is worth noting that the bounds for $\alpha$ become asymptotically `less tight' for small $\gamma<\frac{1}{2}$.
As far as this gap is concerned, we conjecture that the upper bound is the tight one:
the lower bound is obtained via a crude approximation using Jensen's inequality, and this could be potentially tightened (although we haven't been able to do so).
By contrast, the approximation for the upper bound seems less drastic.

We should also note that the results in \cref{THM:WAIT} are driven by the assumption that the arrival of either a client or a provider at every stage of the process is equally likely.
This entails that the expected absolute difference of clients and providers $\abs{\nBoys(\tau) - \nGirls (\tau)}$ can by approximated by a Wiener process as detailed in \cref{app:wait}.
For the latter we know that the expectation is $\sqrt \tau$, so $\abs{\nBoys(\tau)-\nGirls (\tau)} \approx \sqrt \tau$ in expectation.
It would be interesting to consider different arrival processes such as an urn model with delayed replacement where $\abs{\nBoys(\tau)-\nGirls (\tau)}$ would exhibit a different asymptotic behavior;
we leave this analysis to future work.

\if{false}
\textbf{Notes:} 

\textcolor{red}{
\cite{Ash17} cost of waiting is $(t-t_1)+(t-t_2)$ where $t$ is the time of matching, $t_1,t_2$ is the arrival time of the two agents being matched (obviously $t\geq t_1, t_2$).
\\ \\
\cite{Eme16} also treat space and time cost and then minimize their sum.
}
\fi

%% file: Tables/Schedules.tex

\begin{tabular}{>{\flushright\arraybackslash}m{4em}m{30em}}
\hline
$\CS_{\FIFO}$
	&match players as soon as possible on a \acl{FCFS} basis
	\\
\hline
$\CS_{\greedy}$
	&match players as soon as possible
	\\
\hline
$\CS_{\gamma}$
	&match the $k$-th couple when $\Theta(k^\gamma)$ players are on the short side of the market ($0\leq\gamma\leq1$)
	\\
\hline
$\CS_{\patient}$
	&match players optimally after everyone has arrived
	\\
\hline
\hline
$\CS_{\balanced}$
	&match the $k$-th couple when $\Theta(k^{1/2} (\log k)^{1/3})$ players are on the short side of the market.
	\\
\hline
\end{tabular}

%% file: Tables/Results.tex

\begin{tabular}{|l||l|ll|ll|}
\multicolumn{1}{l}{\textsc{Schedule}}
	&\multicolumn{1}{l}{\textsc{Description}}
	&\multicolumn{2}{l}{\textsc{Matching ratio, $\alpha$}}
	&\multicolumn{2}{l}{\textsc{Waiting ratio, $\beta$} }
	\\
\hline
$\CS_{\FIFO}$
	& \ac{FCFS} matching
	& $\Theta(A)$
	&
	&1
	&
	\\
\hline 
$\CS_{\greedy}$
	& Greedy matching 
	& $\Omega(A^{1/2})$
	&\tikzmark{a}{}
	& 1
	&\tikzmark{c}{}
	\\ 
$\CS_{0\leq\gamma<1/2}$
	& Subcritical rate matching
	& $\Omega(A^{1/2-\gamma})$
	&
	& $\Theta(1)$
	&
	\\ 
$\CS_{\gamma=1/2}$
	& Critical rate matching 
	& $\Theta(\log A)$
	&
	& $\Theta(1)$
	&
	\\ 
$\CS_{1/2<\gamma\leq1}$
	& Supercritical rate matching 
	& $\Theta(1)$
	&
	& $\Theta(\tau^{\gamma-1/2})$
	&
	\\ 
$\CS_{\patient}$
	& Patient matching 
	& 1
	&\tikzmark{b}{}
	& $\Theta({\tau}^{1/2})$
	&\tikzmark{d}{}
	\\
\hline 
$\CS_{\balanced}$
	& Balanced matching 
	& $\Theta((\log {A})^{1/3})$
	&
	& $\Theta((\log {A})^{1/3})$
	&
	\\
\hline
\end{tabular}
\link{b}{a}
\link{c}{d}

%% file: 06_balancedschedule.tex

Until now, we have analyzed \aclp{CS} based on the trade-off between waiting time and matching cost, but without explicitly comparing the two.
In this section, we shall commit to a specific class of utility functions in order to make
an explicit comparison between these otherwise
incomparable quantities.

To that end, let $u(\cdot)$ denote the expected utility (or `welfare') of the social planner given a specific \aclp{CS}.
Assume further that the functions expressing this utility depend on both the expected matching cost and the expected waiting time via the additively separable expression
\begin{equation}
u(\CS)
	= u_{\cost}(\alpha_{\CS})+u_{\wait}(\beta_{\CS})
\end{equation}
In order to make comparisons between the utility components $u_{\cost}$ and $u_{\wait}$ we shall first consider their respective maximum values.
It is then natural to assume that $u_{\cost}$ is maximal for the patient \acl{CS} (which minimizes matching cost) and that $u_{\wait}$ is maximal for the  greedy \acl{CS} (which minimizes waiting time).
We shall thus assume that the two maxima are of the same order, viz.,
\begin{equation}
\label{eq:plannerbalance1}
\sigma \cdot u_{\cost}(\alpha_{\patient})
	= (1-\sigma) \cdot u_{\wait}(\beta_{\greedy})
\end{equation}
where $\sigma\in (0,1)$ is a constant factor that specifies the relative importances of the disutilities from mismatching versus waiting.
Naturally, we require that the social planner
seeks to minimize both the costs of matching and the agents' waiting time.
As such, we make the assumption that $u_{\cost}$ is a concave function that decreases in the expected matching cost, and $ u_{\wait}$ is a concave function that decreases in the expected waiting time. 

In view of all this, a social planner is said to be \emph{balanced} if the disutilities from mismatching and waiting display similar growths for large $\tau$.
That is, for a given \acl{CS} $\CS$ with expected matching ratio $\alpha$ and expected waiting time $\beta$, we assume that
\begin{equation}
\label{eq:sameorderutility}
u_{\cost}(\alpha_\CS)
	= \Theta (u_{\wait}(\beta_\CS)) \: \:  \text{ whenever } \:\:  \alpha_{\CS}=\Theta(\beta_{\CS})
\end{equation}

In this general context, we obtain the following result governing balanced social planning:

\begin{theorem}
\label{prop:socialplanner}
Let $\CS_{\balanced}$ be the \acl{CS} that matches the $k$-th couple when  $N-(k-1)=\ceil{k^{1/2} (\log k)^{1/3}}$ players are on the short side of the market.
The expected matching  and waiting  ratios incurred by $\CS_{\balanced}$ are both $\Theta((\log A)^{1/3})$;
moreover, any other schedule $\CS_{f}$ achieving this balance has $f(k) = \Theta(k^{1/2} (\log k)^{1/3})$.
\end{theorem}

\begin{remark*}
For technical reasons we state our result in terms of the number of matched couples ($A=A(\tau)$). Note that, for any clearing schedule where the proportion of matched players increases over time (more precisely, where $\lim_{\tau\to\infty}\frac{A(\tau)}{\nBoys(\tau) + \nGirls(\tau)}= 1$),  $A$ is growing at the same rate as $\tau$.
\end{remark*}

\begin{proof}[Proof of \cref{prop:socialplanner}.]
Consider a \acl{CS} of the form $\CS_{f}$ that matches the $k$-th couple when 
$\ceil{f(k)}$ players on the short side of the market.
In order to balance the expected matching and waiting ratios, any such \acl{CS} would have to satisfy $f(k)=\omega(\sqrt k)$;
otherwise, the expected matching ratio would dominate asymptotically the expected waiting ratio (see \cref{tab:results}).
Thus, without loss of generality, we can assume that $f(k)$ is non-decreasing for large $k$.

Let $t(k,f(k))$ be the stopping time for the event that for the $k$-th time at least $f(k)$ clients and $f(k)$ providers are in the market, assuming that every time this is the case one client and one provider are removed. Finally, recall that $S_\tau= N_\mathcal C (\tau) -N_\mathcal P(\tau)$ is the difference of clients and providers who have arrived to the market until $\tau$.

We begin with the expected matching ratio. For the upper bound we have:
\begin{flalign}
\exof*{\sum_{k=1}^{A} \frac{1}{(\ceil{f(k)}+|S_{t(k,f(k))}|)\ceil{f(k)}}}
	&= \sum_{k=1}^{A} \exof*{\frac{1}{(\ceil{f(k)}+|S_{t(k,f(k))}|)\ceil{f(k)}}}\notag\\
	&\leq \sum_{k=1}^{A} \ceil{f(k)}^{-2}
\end{flalign}
For the lower bound, an algebraic argument which we make precise in \cref{app:match} (\cf  \cref{qwerfwe}) shows that $\exof{t(k, f(k))} < 10k$.
Furthermore, note that $\exof{S_t}$ is strictly increasing in $t$. Thus, by Jensen's inequality,
and \cref{lemma_randomwalkonZdistancefromzero} (which is bounding $|S_t|$ via a combinatorial argument and using Stirling's formula), we get:
\begin{flalign}
\sum_{k=1}^{A}
	\exof*{\frac{1}{(\ceil{f(k)}+|S_{t(k,f(k))}|)\ceil{f(k)}}}
	&\geq \sum_{k=1}^{A} \frac{1}{(\ceil{f(k)}+\mathbb E [|S_{10k}|])\ceil{f(k)}}
	\notag\\
	&\geq \frac{\pi}{\sqrt{2}e } \sum_{k=1}^{A} \frac{1}{(\ceil{f(k)}+\sqrt{10k})\ceil{f(k)}}
	\notag\\
 &	= \Theta\parens*{\sum_{k=1}^{A} \frac{1}{f(k)^{2}}}
\label{helphelpehelp}
\end{flalign}
where the last line follows from the assumption $f(k)=\omega(k^{1/2})$. 
Thus the two bounds together with the fact that the patient schedule has finite matching cost yield the result that, for $f(k)=\omega(k^{1/2})$ the expected matching ratio is 
\(
\alpha(A)
	= \Theta\parens*{\sum_{k=1}^{A} 1/\bracks{f(k)^{2}}}.
\)
	
We proceed, by considering the incurred waiting time.
Recall that $N(\tau) - A(\tau)$ is the number of agents on the shorter side of the market at time $\tau$, so
\(
N-A
	= \ceil{f(k-1)}-1
\)
after the $(k-1)$-st match.
The number of clients and the number of providers that need to arrive to the market before the $k$-th match is thus upper bounded by
\begin{equation}
\ceil{f(k)}-\big(\ceil{f(k-1)}-1\big )=1+\ceil{f(k)}-\ceil{f(k-1)}
\leq 2
\end{equation}
where the last inequality follows from the fact that $f(k)=o(k)$ is a necessary condition for a feasible \acl{CS} (that is, a \acl{CS} where the proportion of unmatched versus matched players is decreasing).
The expected waiting time accrued between the $(k-1)$-st and the $k$-th match is therefore upper bounded by:
\begin{equation}
\label{biggy}
\Delta_k:=\underbrace{\mathbb E [\text{time s.t. $\geq 2$ clients \& $\geq 2$ providers enter market}]}_{=:\Delta_k^1}\cdot \underbrace{\big(2\ceil{f(k)} +\ceil{g(k)} \big)}_{=:\Delta_k^2}
\end{equation}
where $\ceil{g(k)}$ is a function that we will use to upper bound $|S_k|$, viz., the random variable constituting  the absolute difference of clients and providers in the market at time $\tau(k)$.
For posterity, note also that $\Delta_k^1$ is the expectation of the time between the  $(k-1)$-th and the $k$-th match and $\Delta_k^2$
provides an upper bound for the number of agents waiting in the time interval between the $(k-1)$-th and the $k$-th match.

Given the arrival of agents is governed by a Poisson clock of rate one, we have $\Delta_k^1=5$, \ie on average, five agents need to enter the market to have at least two clients and at least two providers.
To see this, let $Y$ be the number of flips of a coin required to observe at least 2 heads (clients) and 2 tails (providers). The event `$Y>k$' is then equivalent to the union of the events `$\binom{k}{k-1}$ heads' and `$\binom{k}{k-1}$ tails'. The two latter events are disjoint and each has probability $\frac{k}{2^k}$. Thus $\mathbb P [Y>k]=\frac{k}{2^{k-1}}$ and we have
\begin{eqnarray}
&\exof{Y}
	&= \sum_{k=0}^\infty \probof{Y>k}
	= 1 + 2\sum_{k=1}^\infty \frac{k}{2^{k}}
	= 1 + 2\sum_{k=1}^\infty \sum_{j=1}^{k} \frac{1}{2^{k}}\\
&&= 1 + 2\sum_{j=1}^\infty \sum_{k=j}^\infty \frac{1}{2^{k}}
	= 1 + 2\sum_{j=1}^\infty \frac{1}{2^{j-1}}
	= 5
\end{eqnarray}
We thus have for \cref{biggy} 
\begin{equation}
\Delta_k=5\cdot \big(2\ceil{f(k)} +\ceil{g(k)} \big)
\end{equation}
Next, to choose the function $g(k)$, note that the law of the iterated logarithm gives
\begin{flalign}
\lim_{k\to\infty} \frac{|S_k|}{\sqrt{2k\log\log k}}
	= 1.
\end{flalign}
Hence, by choosing $g(k) = \sqrt{2k\log\log k}$, the random variable $|S_k|$ is asymptotically bounded from above by $g(k)$ with probability one.

We consider two cases below, which are exhaustive by \citet[Theorem, page 18]{Har10old}:

\paragraph{Case 1: $f(k) = \Omega(g(k))$.}

For the first case we have:
\begin{equation}
5\cdot (2f(k) + g(k))
	= \Theta(f(k))
\end{equation}
The expected waiting ratio until $A$ pairs have been matched is bounded from above by
\(
A^{-3/2} \sum_{k=1}^A\Theta(f(k)),
\)
where we are using the fact that the expected waiting time for the greedy schedule is given by $A^{3/2}$ (see \cref{prop:optimalbaseline}).
A trivial lower bound for the expected waiting ratio is then given by
\begin{equation}
\frac{1}{A^{3/2}}\sum_{k=1}^A 2 {f(k)}=\frac{1}{A^{3/2}}\sum_{k=1}^A  \Theta(f(k))
\end{equation}
Thus, the expected waiting ratio is given by
\begin{equation}
\beta(A)
	= \Theta\parens*{\frac{1}{A^{3/2}}\sum_{k=1}^A  f(k)}
\end{equation}

Moving to the comparison of matching and waiting ratios, we recall that $u_{\cost}$ and $u_{\wait}$ are decreasing and concave and are of the same order (by assumption). Thus $u=u_{\cost}+u_{\wait}$ is maximized if and only if $\alpha = \Theta(\beta)$.
In turn, this holds if and only if
\begin{flalign}
\label{eqeqeqeqeq}
\sum_{k=1}^{A} \frac{1}{f(k)^{2}}
	&= \Theta \parens*{ \frac{1}{A^{3/2}}\sum_{k=1}^A  f(k) }
\intertext{or, equivalently, if and only if}
\label{eqnice}
\int_1^{A} \frac{1}{f(x)^2} \dd x
	&= \Theta \parens*{ \frac{1}{A^{3/2}}\int_1^{A} f(x) \dd x }
\end{flalign}
where the asymptotic passage from summation to integration \textendash\ \ie from \cref{eqeqeqeqeq} to \cref{eqnice} \textendash\ is made precise in \cref{app:tech}.

We shall show that $f(x)=\Theta(\sqrt x (\log x)^{1/3})$ is the unique solution to \cref{eqnice} up to order.
To simplify notation, let $f(x)=\sqrt x (\log x)^{1/3}$, so the \ac{LHS} of \cref{eqnice} becomes    
\begin{equation}
\int_1^A \frac{1}{x(\log x)^{2/3}}
	= 3(\log A)^{1/3} + c
\end{equation}
where $c$ is uniformly bounded and independent of $A$.
Next, focusing on the \acs{RHS} of \cref{eqnice}, we get
\begin{flalign}
\frac{1}{A^{3/2}}\int_1^A \sqrt x (\log x)^{1/3} \dd x
	&= (\log A)^{1/3} - \frac{1}{A^{3/2}}\int_1^A x^{3/2} \frac{1}{3x (\log x)^{2/3}} \dd x
	\notag\\
	&= (\log A)^{1/3} -\frac{1}{A^{3/2}} \underbrace{\int_1^A \sqrt x \frac{1}{3 (\log x)^{2/3}} \dd x}_{= o\big(\int_1^A \sqrt{x} \cdot (\log x)^{1/3} \dd x \big)}
	\notag\\
	&= \Theta((\log A)^{1/3})
\end{flalign}
Our uniqueness claim follows by noting that the \ac{LHS} of \cref{eqnice} is decreasing in $f(x)$ (in orders of magnitude of the upper bound of the integral) while the \ac{RHS} is increasing in $f(x)$.

\paragraph{Case 2: $f(k) = o(g(k))$.}
For the second  case, assume that $f(k) = o(g(k))$.
This implies for the matching cost that%
\footnote{Formally, for $\tau \geq e$, the integrand is not well-defined, but the Cauchy principal value of the integral remains finite, and this is the value we are using for $\tau\leq e$.
This issue could be side-stepped by shifting the lower limit of the integral to a higher value, but we do not do so in order to simplify the presentation.}
\begin{equation}\label{asdfadsvda}
\int_{1}^{A} \frac{1}{f(\tau)^2} \dd\tau
	= \omega\parens*{\int_{1}^{A} \frac{1}{g(\tau)^{2}} \dd\tau}
\end{equation}
The integral on the \ac{RHS} of \cref{asdfadsvda} can then be bounded from below as follows
	\begin{equation}\label{asdfadsvdas}
\int_{1}^{A} \frac{1}{g(\tau)^{2}} \dd\tau
	= \int_{1}^{A} \frac{1}{4\tau\log\log\tau} \dd\tau
	= \omega\parens*{ \int_{1}^{A} \frac{1}{4 \tau (\log \tau)^{2/3}} \dd\tau }
\end{equation}
For the integral on the \ac{RHS} of \cref{asdfadsvdas} we have
\begin{equation}
\int_{1}^{A} \frac{1}{4 \tau (\log \tau)^{2/3}}d\tau
 	= \Theta \big ((\log A)^{1/3}\big ),
\end{equation}
Hence, combining these last approximations, we finally get
\begin{flalign}
\int_{1}^{A} \frac{1}{f(\tau)^{2}} \dd\tau
	&= \omega\big( (\log A)^{1/3} \big).
\end{flalign}
Thus any solution satisfying $f(k) = o(g(k))$ (Case 2) has expected matching cost that is $\omega(1)$ relative to the optimal solution.
This completes the proof that $f(k)=\Theta(\sqrt x (\log x)^{1/3})$ is the unique optimal \aclp{CS} for the balanced social planner.
\end{proof}

Note that, up to logarithmic factors, the balanced \acl{CS} is  close to the \acl{CS} $\CS_{\gamma=1/2}$ which signifies a first-order phase transition for the expected matching ratio.
As discussed earlier, $\CS_{\gamma=1/2}$ only signifies a second-order phase transition for the expected waiting ratio, thus explaining the gap between $\CS_{\balanced}$ and $\CS_{\gamma=1/2}$.
In practice however, $\CS_{\gamma=1/2}$ seems to be a reasonable approximation for a balanced social planner.

%% file: 06b_generalizations.tex

So far, we focused our attention on dynamic clearing games with exponentially distributed match costs. We shall show below that the developed techniques can also be used to study a more abstract model. 
Rather than modeling match costs directly by defining the distribution of each potential match cost, $w_{ij}$, we take a macroscopic viewpoint and posit that the cost of matching a couple 
depends on the number of clients and providers currently in the market. This cost
may be the expected cost of matching the cheapest couple or a cost associated to market making more generally.
In practice, this cost can be learned from past data on clearing events. We shall thus focus our generalization on identifying breaking points and their respective consequences for different cost regimes. 

Write $\NBoys(\tau)= \nBoys(\tau)-A(\tau)$ for the number of clients in the market at time $\tau$ and $\NGirls(\tau)= \nGirls(\tau)-A(\tau)$ for the number of providers respectively.
Then, the expected cost can be written w.l.o.g. as
\begin{equation}
\g(\NBoys,\NGirls)
\end{equation}
where $\g\from \mathbb R^+ \times \mathbb R^+ \to \mathbb R^+$ is a non-increasing function (in either argument).\footnote{We define the function $\g$ on the real numbers, but note that it is only the values on $\mathbb N \times \mathbb N$ which enter the analysis of \aclp{CS}.}
Intuitively, $\g$ determines how the expected minimum cost of matching decreases as more players coexist in the market.
For example, if $\g(\NBoys,\NGirls)=\frac{1}{\NBoys \cdot \NGirls}$, we revert to the previous analysis resulting from exponentially distributed match costs. This is the case since the expected minimum of $\NBoys \cdot \NGirls$ independent $\exp(1)$ random variables is equal to $\frac{1}{\NBoys \cdot \NGirls}$.

In view of this, it stands to reason that the asymptotic behavior of the market will be captured by the rate at which the expected minimum matching cost $\g(\NBoys,\NGirls)$ vanishes as a function of $\NBoys,\NGirls\to\infty$.
\Cref{prop:general} below makes this intuition precise and identifies a specific threshold beyond which it \emph{is} possible to get a `free lunch'. We restrict our analysis to the case $\delta>1$ to guarantee that the patient \acl{CS} has finite expected matching cost, \ie $\cost_{\patient}<\infty$.

\begin{theorem}
\label{prop:general}
Suppose that the expected minimum matching cost decays as $\g(\Theta(x),\Theta(x)) = \Theta(1/x^{\delta})$
for some $\delta>1$.%
Then:
\begin{enumerate}
[\quad\upshape(\itshape i\hspace*{1pt}\upshape)]
\item
For $1<\delta\leq 2$ there is no `free lunch'. In particular, the critical rate clearing schedule, that is, the clearing schedule with expected matching ratio $\Theta(\log (A))$, is given by
$\CS_{\gamma=1/\delta}$.
\item
For $\delta> 2$, `free lunch' exists. In particular, the \aclp{CS}  $\CS_\gamma$ with $\gamma\in (\frac{1}{\delta},\frac{1}{2}]$ guarantee that the expected matching and waiting ratios are both finite.
\end{enumerate}
\end{theorem}

\begin{proof}[Proof of \cref{prop:general}.] 
We first consider the upper bound.
Given $g(\Theta(x),\Theta(x))=\Theta\left(\frac{1}{x^\delta}\right )$ and since $\g$ is increasing in both arguments we have:
\begin{flalign}
 \exof*{\sum_{k=1}^{A} {g(\NBoys,\NGirls)}\given \min\{\NBoys,\NGirls\}=\floor{k^{1/\delta}} }
& \leq   \sum_{k=1}^{A} {g(\floor{k^{1/\delta}},\floor{k^{1/\delta}})}
\notag  \\
&= \Theta\left( \sum_{k=1}^{A}\frac{1}{k}\right)   \:= \: \Theta(\log A )
\end{flalign}
where the last inequality follows from the bounds for the harmonic series.
Thus, given the optimal clearing schedule has finite matching  cost, the expected matching ratio is smaller than $O(\log A+1)$.

For the lower bound note that $\mathbb E [t(k,k^{1/\delta})]<10k$ by similar arguments as in  \cref{qwerfwe} and by recalling that $\delta> 1$. 
Thus,  combining Jensen's inequality ($\frac{1}{x}$ is convex) with Markov's inequality, \cref{lemma_randomwalkonZdistancefromzero}, and recalling that $\delta\leq 2$ we get:

\begin{flalign}
&\sum_{k=1}^{A} \exof*{ {g(\NBoys,\NGirls) }\given \min\{\NBoys,\NGirls\}=\floor{k^{1/\delta}}}
 \notag\\
&=\sum_{k=1}^{A}  \exof*{ {g(\NBoys,\NGirls) }\given \min\{\NBoys,\NGirls\}=\floor{k^{1/\delta}},\: t(k,k^{1/\delta}) <20k} \cdot \mathbb P \left[t(k,k^{1/\delta}) <20k\right]
 \notag\\
&+ \sum_{k=1}^{A}   \exof*{ {g(\NBoys,\NGirls) }\given\min\{\NBoys,\NGirls\}=\floor{k^{1/\delta}},\: t(k,k^{1/\delta}) >20k}\cdot \mathbb P [t(k,k^{1/\delta}) >20k]
 \notag\\
 &\geq \sum_{k=1}^{A}  \exof*{ {g( \NBoys,\NGirls)}\given \min\{\NBoys,\NGirls\}=\floor{k^{1/\delta}},\: t(k,k^{1/\delta}) <20k} \cdot  \frac{1}{2}
 \notag\\
	&=\Theta\left ( \sum_{k=1}^{A}   {g( \floor{k^{1/\delta}},\floor{k^{1/\delta}}) }  \right )
	= \Theta\left( \sum_{k=1}^{A}\frac{1}{k}\right) = \Theta( \log A )
\end{flalign}
where we used that $\NBoys=\Theta( \NGirls)$  given 
$|S_t|=|\nBoys(t) - \nGirls(t)|=|\NBoys (t)-\NGirls (t)|=O(\sqrt{20k})$ since we are in the case
$t(k,k^{1/\delta}) <20k$ and by the assumption that for all $\lambda, \mu$ we have $\g(\lambda\cdot x,\mu\cdot x)=\Theta(x^\delta)$.
Thus, given the optimal clearing schedule has finite matching cost, the expected matching ratio is  $\Omega (\log A)$, concluding the proof together with the upper bound. 

To summarize, for $\delta\in (1,2]$  the critical rate \acl{CS} is given by $\CS_{\gamma=1/\delta}$. Thus the critical rate clearing schedules are given by $\CS_\gamma$ with $\gamma\in(0,1/2]$ and by \cref{THM:WAIT}, the expected waiting ratio for these schedules is not finite. Note that the case $\delta=2$ is simply \cref{thm:freelunch}. We conclude that there is no `free lunch'.

 (2) By \cref{THM:WAIT} the expected waiting ratio is finite for all \aclp{CS} $\CS_{\gamma}$ with $\gamma\leq\frac{1}{2}$. 
 
 We upper bound the expected matching ratio for the clearing schedule $\CS_\gamma$ for $\gamma>\frac{1}{\delta}$:
 For the upper bound we have with $g(\Theta(x),\Theta(x))=\Theta\left(\frac{1}{x^\delta}\right )$:
\begin{flalign}
 \exof*{\sum_{k=1}^{A} {g(\NBoys,\NGirls)}\given \min\{\NBoys,\NGirls\}=\floor{k^{1/\delta}}}
& \leq   \sum_{k=1}^{A} {g(\floor{k^{1/\delta}},\floor{k^{1/\delta}})}
\notag  \\
&= \Theta\left( \sum_{k=1}^{A}\frac{1}{k^{\gamma\cdot \delta }}\right)   \:= \: \Theta(1 )
\end{flalign}
where the last identity holds since  $\gamma>\frac{1}{\delta}$.
 
Thus,  for given $\delta$ the \aclp{CS} $\CS_\gamma$ with $\gamma\in (\frac{1}{\delta},\frac{1}{2}]$ guarantee that the expected matching ratio and waiting ratio are both finite, \ie free lunch.
\end{proof}


\cref{prop:general}(\emph{i}) extends our previous analysis by showing how the critical rate clearing schedule moves dependent on $\delta$. In fact \cref{prop:general}(\emph{ii}) shows that the conclusion of  \cref{thm:freelunch} does not hold for the regime $\delta>2$, that is, there exists a free lunch and in particular it is achieved by the \aclp{CS}  $\CS_\gamma$ with $\gamma\in (\frac{1}{\delta},\frac{1}{2}]$. 
This is because for quickly decaying matching costs it becomes easier to choose a `good' schedule, and thus the `window of opportunity' is increasing in the derivative of $g$.
One can build intuition for this result by reasoning about market settings that differ in terms of match cost variability: in markets where match costs are generally rather similar, thickening the market by waiting will only lead to a meaningful positive effect in terms of expected match cost reduction when waiting for a long time. By contrast, when match costs vary substantially, match costs reduce in expectation with much less delay, thus making it more likely for a mechanism designer to get a free lunch.
Importantly, the \acl{CS} $\CS_{\greedy}$ is never optimal when dealing with many types, independent of the match cost distribution at hand. 

%% file: 07_discussion.tex
In this paper, we studied the \emph{dynamic clearing game}, where heterogeneous clients and providers arrive uncoordinatedly in order to be matched. 
We studied the trade-off a social planner is facing between two competing objectives: a) to reduce players' \emph{waiting time} before getting matched; and b) to form efficient pairs in order to reduce \emph{matching cost}.

Our analysis of the dynamic clearing game reveals that a multi-objective social planner often faces a substantial trade-off. Starting with the micro-founded model for match costs we showed that there exists no free lunch, that is, there is no clearing schedule that is approximately optimal in terms of both waiting time and matching cost. We identified a unique breaking point where a stark reduction in matching cost compared to a stark increase in waiting cost occurs. 
In line with recent works by \citet{Ash17}, \citet{Ash18b} and many others, we focused on
a concrete class of social welfare functions that weigh costs from waiting versus matching on a comparable scale and identify the optimal \acl{CS}, namely, the \acl{CS} that matches the $k$-{th} couple when $\Theta(\sqrt k (\log k)^{1/3})$ players are on the short side of the market. 

Generalizing the model, we abstract away from modeling match costs directly and take a macroscopic viewpoint. Positing that the cost of matching a couple depends on the number of clients and providers currently in the market we identify two regimes. One, where no free lunch continues to hold, the other, where there is a window of opportunity to be optimal along both dimensions, that is free lunch.

There are multiple directions in which our analysis could be extended. Perhaps the most evident avenue for future research is to model market participation behavior game-theoretically, which would lead to new strategic considerations and probably induce other matchings (see, \eg \citealt{Bac18}). This analysis could be pursued in more applied contexts, for instance relating to our motivating example of a labor market with a central employment bureau, where waiting costs could be interpreted as benefits payable by the bureau. An unemployed worker might forgo some of these benefits by (repeatedly) rejecting matches. This is the case because longer waiting, even though borne out of strategic behavior, may improve the match quality (reducing matching cost).

A second route for further investigation is to enlarge the options of the social planner in terms of clearing schedules. For one, the social planner could be learning from market observations about the distribution of match costs, which incidentally we may also allow to follow other, more general classes of distributions. This would allow the social planner to formulate more sophisticated clearing schedules that incorporate match costs between players that are currently in the market. In particular, if the social planner learns that a given agent may be `hard to match', then it might be sensible to match that agent directly and not incur further waiting cost. Furthermore, the social planner may want to match more than one couple at a time. 

The study of dynamic market institutions is clearly fascinating, with tremendous scope for progress in (old and new) applications, where research has only just started.
Our contribution has been to go beyond binary match values, and to identify breaking points under incomplete information. 
We hope that our framework is able to provide fertile ground for further research, both theoretical and applied to real-world market contexts, in particular as regards thinking about whether the kinds of breaking points we describe are relevant in the optimal design of such markets.


%% file: 10_appendix_match.tex

Before turning to the proof we introduce the following definition and lemmas.

\begin{definition}
 Let $X_1,X_2, \ldots$ be iid random variables with $\mathbb P [X_i=1]=\mathbb P [X_i=-1]=\frac{1}{2}$.
 \begin{itemize}
 \item Let $S_k=\sum_{i=1}^k X_i$.
 \item Let  $t(k,C)$ be the stopping time  for the event that for the $k$-th time at least $C$ clients and $C$ providers are in the market, assuming  that every time this is the case one client and one provider are removed. 
\end{itemize}
\end{definition}

\begin{lemma}
\label{lemma_minofexpvariables}
Let $w_{ij}\sim \exp(\rate_j)$ for $j=1,2,\dotsc,N$, be a family of independent exponentially distributed random variables.
Then 
\begin{equation}
\min \{w_{i1},w_{i2},\ldots,w_{iN}\}
	\sim \exp\parens*{\sum_{j=1}^N \rate_j}.
\end{equation}
In particular, if for all $j$, $\rate_j=1$, then $\mathbb E [ \min_j w_{ij}]=\frac{1}{N}$.
\end{lemma}

\begin{proof}
This proof is standard but we repeat it for the sake of completeness. The random variable $w_{ij}$  has cumulative distribution function
\begin{equation}
F_{w_{ij}}
	= \probof{w_{ij}\leq x}
	= 1-e^{-\rate_j x}
	\quad
	\text{for all $x>0$ and all $j=1,2,\dotsc,N$.}
\end{equation}
Now, define the random variable $Y=\min\{w_{i1},w_{i2},\ldots,w_{iN}\}$.
Then, the cumulative distribution function of $Y$ is
\begin{flalign}
F_Y(y)
	&= \probof{Y\leq y}
	\notag\\
	&= 1 - \probof{Y\geq y}
	\notag\\
	&= 1 - \probof{\min\{w_{i1},w_{i2},\ldots,w_{iN}\}\geq y}
	\notag\\
	&= 1 - \probof{w_{i1}\geq y} \cdot \probof{w_{i2}\geq y}\cdot \ldots \cdot \probof{w_{iN}\geq y}
	\notag\\
	&= 1 - e^{-\rate_1 y} \cdot  e^{-\rate_2y} \cdot \dotsm \cdot  e^{-\rate_Ny}
	\notag\\
	&= 1-e^{-\sum_{j=1}^N\rate_j y}\:\:\:\:\:\: y>0
\end{flalign} 
The latter cumulative distribution function is that of an exponential variable with parameter $\sum_{j=1}^N\rate_j$.
\end{proof}


\begin{lemma}
\label{lemma_randomwalkonZdistancefromzero} For $S_k$  defined as above we have:\footnote{Note that $\lim_{k\to\infty} \mathbb E [|S_k|]=\sqrt{\frac{2}{\pi}}\cdot \sqrt {k}$ \citep{Pet56}.}
\begin{equation}
0.67\cdot  \sqrt k \: \lessapprox \: \frac{2\pi}{e^2}\cdot \sqrt{\frac{2}{\pi}}\cdot \sqrt {k}\:\:
	\leq \exof{|S_k|}
	\leq \frac{e}{\sqrt{\pi}} \cdot \sqrt{\frac{2}{\pi}}\cdot \sqrt {k}
	\lessapprox 1.23 \cdot \sqrt k
\end{equation}
\end{lemma}

\begin{proof}
The starting point of our proof is an intermediate result in the proof of the limit of the expected absolute value of the 1-d random walk, which is detailed in \citet[Equations 29a and 29b]{Hiz11} and is based on combinatorial arguments via the binomial distribution:
\begin{flalign}
\mathbb E[|S_k|]
	&
	= \begin{cases}
		\displaystyle
		\frac{1}{2^{k-2}}\frac{k}{2}{\binom{k-1}{k/2}}
		= \frac{k}{2^{k}}\frac{k!}{[(k/2)!]^2}
			&\quad
			\text{for $k$ even},
		\\[2em]
		\displaystyle
		\frac{1}{2^{k-1}}\frac{k+1}{2}{\binom{k}{(k+1)/2}}
		= \frac{k+1}{2^{k+1}}\frac{(k+1)!}{[((k+1)/2)!]^{2}}
			&\quad
			\text{for $k$ odd}.
	\end{cases}
\end{flalign}
Since $\mathbb E[|S_{2k}|]=\mathbb E[|S_{2k-1}|]$ it  suffices to analyze the case where $k$ is even. 
To that end, we will use Stirling's formula to bound $k!$ from above and below as
\begin{equation}
\sqrt{2\pi}\cdot k^{k+1/2}\cdot e^{-k}
	\leq k!
	\leq e \cdot k^{k+1/2} \cdot e^{-k}
\end{equation}
For $k$ even, we may bound $\abs{S_k}$ from above as:
\begin{flalign}
\mathbb E[|S_k|]=
	\frac{k}{2^{k}} \frac{k!}{[(k/2)!]^2}
	&\leq \frac{k}{2^{k}} \frac{e \cdot k^{k+\frac{1}{2}} \cdot e^{-k}}{{2\pi}\cdot (k/2)^{k+1}\cdot e^{-k}}
	= \frac{e}{\sqrt{2\pi}} \cdot \sqrt{\frac{2}{\pi}}\cdot \sqrt k
\end{flalign}
Next, we lower bound $|S_k|$ for $k$ even:
\begin{flalign}
\mathbb E[|S_k|]
	= \frac{k}{2^{k}}\frac{k!}{[(k/2)!]^2}
	&\geq \frac{k}{2^{k}} \frac{\sqrt{2\pi}\cdot k^{k+1/2} \cdot e^{-k}}{e^2 \cdot (k/2)^{k+1} \cdot e^{-k}}
	= \frac{2\pi}{e^2}\cdot \sqrt{\frac{2}{\pi}}\cdot \sqrt k
\end{flalign}
This concludes  the proof for $k$ even. For $k$ odd we have with the observation that $|S_k|=|S_{k+1}|$:
\begin{subequations}
\begin{alignat}{2}
\frac{2\pi}{e^2}\cdot \sqrt{\frac{2}{\pi}}\cdot \sqrt {k}
	&\:\leq \frac{2\pi}{e^2}\cdot \sqrt{\frac{2}{\pi}}\cdot \sqrt {2\ceil{k/2}}
	&\:\:\leq \exof{\abs{S_{k+1}}}
	&\:\:=\exof{\abs{S_{k}}}
\intertext{and}
\frac{e}{\sqrt{\pi}} \cdot \sqrt{\frac{2}{\pi}}\cdot \sqrt{k}
	&\:\geq \frac{e}{\sqrt{\pi}} \cdot \sqrt{\frac{2}{\pi}}\cdot \frac{1}{\sqrt 2} \cdot \sqrt {2\ceil{k/2}}
	&\:\:\geq \exof{\abs{S_{k+1}}}
	&\:\:=\exof{\abs{S_k}}
\end{alignat}
\end{subequations}
\end{proof}

For the sake of limiting notation the proposition and proof are stated for the \aclp{CS} with $f(k)=\ceil{k^\gamma}$ rather than for $\Theta (f)$. Adding constant upper and lower bounds is straightforward and thus omitted. Recall that $t(k,f(k))$ is the stopping time for the event that for the $k$-th time at least $f(k)$ clients and $f(k)$ providers are in the market, assuming that every time this is the case one client and one provider are removed.

\begin{proof}[Proof of \cref{THM:COST}.] 
Throughout the proof we shall simplify notation by omitting the fact that some of the matching schedules are defined via the ceiling of functions mapping to $\mathbb R^+$ (e.g., $\ceil{k^\gamma}$). The results are not changed by the omission since match costs are never underestimated and overestimated by very little. Further,  while they are stated together in the proposition, we study the  \aclp{CS} $\CS{\gamma=0}$ and $\CS_{0<\gamma<1/2}$ separately since they require different arguments.

\Paragraph{First come, first served \textpar{$\CS_{\FIFO}}$} In \ac{FCFS}  the cost of each match is the expectation of a single match cost, that is $\rate$. After $A$ matches have occurred, the expected incurred cost is $\rate A$. Thus, given the patient \acl{CS} has cost bounded above by $\frac{\pi^2}{6 \underline \rate}$, the expected matching ratio is equal to $\frac{\rate\cdot A}{\pi^2/(6 \underline \rate)}$.

Before stating the proofs for the other results recall from the proof of \cref{prop:optimalbaseline}, that the exponential distribution is closed under scaling. We shall thus simplify notation and assume that for all $i,j$ $w_{ij}\sim \exp(1)$. Note that, for lower bounds the scaling factor $\frac{1}{\overline \rate}$ needs to be applied and for upper bounds the scaling $\frac{1}{\underline\rate}$ needs to be applied. But note that those scaling factors are constant with respect to $\tau$ (and thus $A$) and therefore do not influence the orders of the limiting results.

\Paragraph{Greedy matching \textpar{$\CS_{\greedy}$}} The $k$-th match happens when the minimum of the number of clients and providers who already arrived to the market is $k$, that is, at time $t(k,1)$. The expected weight of the 
$k$-th match depends on the number of players currently present on the long side of the market (since on the short side there is only one agent).
 This random variable is given by $|S_{t(k,1)}|+1$.
By \cref {lemma_minofexpvariables} the expected weight thus is $\mathbb E [\frac{1}{|S_{t(k,1)}|+1}]$. 
The first $A$ matches thus  have an expected cost of 
\begin{equation}
\label{equationgreedy}
\mathbb E[ \sum_{k=1}^A  \frac{1}{|S_{t(k,1)}|+1}].
\end{equation}
Given that we study fixed $A$ (the number of matches that)  
we have:\footnote{Note that $t$ (the total number of client and providers who have arrived to the market) depends  on $A$ (and vice versa). Therefore, \cite{Wal44}'s equation does not apply and thus the route of inquiry to study the matching cost at some continuous time $\tau$ does not  work since we could not   interchange summation and expectation.}
\begin{equation}
\exof*{\sum_{k=1}^{A}  \frac{1}{|S_{t(k,1)}|+1}}
	= \sum_{k=1}^{A} \mathbb E[ \frac{1}{|S_{t(k,1)}|+1}]
\end{equation}
 Next, by Jensen's inequality ($\frac{1}{x}$ is convex) we have:
\begin{equation}
\sum_{k=1}^{A} \mathbb E [\frac{1}{|S_{t(k,1)}|+1}] \: > \: \sum_{k=1}^A\frac{1}{\mathbb E[|S_{t(k,1)}|+1]}
\:=\:\sum_{k=1}^A\frac{1}{\mathbb E[|S_{t(k,1)}|]+1} \label{helpme3}
\end{equation}
We shall now approximate $\mathbb E [t(k,1)]$. By \cref{lemma_randomwalkonZdistancefromzero} we have $\mathbb E [S_t]<1.23\sqrt t$. Thus the short side of the market has 
$\frac{t-1.23\sqrt t}{2}$ agents. Setting $k=\frac{t-1.23\sqrt t}{2}$ and solving the quadratic equation we find the crude upper bound for the expectation:\footnote{ 
We solve
$k=\frac{t-1.23\sqrt t}{2}$.
Setting $t=u^2$ and rearranging we solve  quadratic equation
$ u^2-u-2k\stackrel{!}{=}0$.
The solutions are:
\begin{equation*}
u_{1,2}=\frac{1.23\pm \sqrt {1.23^2+8k}}{2}
\end{equation*}
Given the variable transformation the positive solution is selected.}
\begin{equation}
\label{eqhelpeq}
\mathbb E [t(k,1)]=\big (\frac{1.23+ \sqrt {1.23^2+8k}}{2}\big )^2<\frac{3}{4}+2k+2\sqrt k <5k
\end{equation}
Returning to \cref{helpme3} we have with \cref{lemma_randomwalkonZdistancefromzero}:
\begin{align}
\sum_{k=1}^A\frac{1}{\mathbb E[|S_{t(k,1)}|]+1}
> \sum_{k=1}^A\frac{1}{\mathbb E[|S_{5k}|]+1}
&>\frac{1}{1.23} \sum_{k=1}^{A}\frac{1}{\sqrt {5k} +1}
\\
>\frac{1}{1.23} A\cdot\frac{1}{\sqrt {{5A}}+1}
&>  A\cdot\frac{1}{5\sqrt {{A}}}
= \frac{\sqrt A}{5}
\end{align}
Thus, given the optimal schedule has $\cost_{\patient}(A)\leq \frac{\pi^2}{6 \underline \rate}$, the expected matching ratio is lower bounded by $ \frac{\sqrt {A}}{5\pi^2/(6 \underline \rate)}$.

\Paragraph{Subcritical matching \textpar{$\CS_{\gamma =0}$}}
We shall fix the clearing schedule such that it matches a couple every time some fixed $C\in \mathbb N$ players are on the short side of the market ($N-A=C$) and note that it belongs to the family of \aclp{CS} $\CS_{\gamma=0}$.

 Next, note that $t(k,C)=t(1,C)+t(k-1,1)$, since we assume that every time at least $C$ clients and $C$ providers are in the market exactly one client and one provider match and thus leave the market. Similarly to the proof of \cref{THM:COST}$(\CS_{\greedy})$ we can bound $\mathbb E [t(1,C)]$ from above by noting that it is equal to $\mathbb E [t(C,1)]$. Thus 
 \begin{equation}
\label{qwerfwe}
 \mathbb E [t(k,C)]=\mathbb E [t(C,1)]+\mathbb E [t(k-1,1)] <5C+5k
\end{equation}
 
 
 With the latter and, as above by Jensen's inequality ($\frac{1}{x}$ is convex) and \cref{lemma_randomwalkonZdistancefromzero} we have for the expected matching cost:
\begin{flalign}
\label{hel123}
\exof{\sum_{k=1}^{A} \frac{1}{(C+|S_{t(k,C)}|)C}}
	&\geq \sum_{k=1}^{A}\frac{1}{(C+\mathbb E[|S_{t(k,C)}|])C}
	\geq \sum_{k=1}^{A}\frac{1}{(C+\mathbb E[|S_{5C+5k}|])C}
	\notag\\
	&\geq \frac{1}{1.23}\sum_{k=1}^{A}\frac{1}{(C+\sqrt {5C+5k})C}
	\geq \frac{A}{1.23}\cdot\frac{1}{(C+\sqrt {5C+5A})C}
	\notag\\
	&= \frac{1}{1.23\cdot C}\cdot\frac{A-C}{\sqrt {5C+5A}+C}
	= \Omega (\sqrt A)
\end{flalign}
Thus, given the optimal clearing schedule has finite cost the expected matching ratio is $\alpha(A) =\Omega (\sqrt {A})$.

The second part of the assertion follows by observing:
\begin{equation}
 \mathbb E [\sum_{k=1}^{A} \frac{1}{(C+|S_{t(k,C)}|)C}] 
	< \frac{1}{C}\mathbb E [\sum_{k=1}^{A} \frac{1}{1+|S_{t(k,1)}|}]
\end{equation}

\Paragraph{Subcritical matching \textpar{$\CS_{0<\gamma < 1/2}$}}
For the upper bound, by  \cref{lemma_minofexpvariables}, we have:
\begin{align}
\exof{\sum_{k=1}^{A} \frac{1}{( k^\gamma+|S_{t(k,\sqrt k)}|)k^\gamma}}
	&\leq \sum_{k=1}^{A}\frac{1}{k^{2\gamma}} 
	= 1+\sum_{k=2}^{A}\frac{1}{k^{2\gamma}}
	\leq 1+\int_{x=1}^{A}\frac{1}{x^{2\gamma}} \dd x
	\notag\\
	&= 1+\bracks*{\frac{1}{1-2\gamma}x^{1-2\gamma}}_{x=1}^{A}
	\leq 1+  \frac{1}{1-2\gamma} A^{1-2\gamma}
\end{align}
Thus, given the optimal clearing schedule has finite cost, the expected matching ratio is $\alpha (A) =\bigoh(A^{1-2\gamma})$ for $0<\gamma<\frac{1}{2}$.


For the lower bound,   note that $t(k, k^\gamma)<10k$ for $\gamma <1$ by similar arguments as in  \cref{qwerfwe}. Further note that $\mathbb E[|S_t|]$ is strictly increasing in $t$. Thus,  with Jensen's inequality ($\frac{1}{x}$ is convex):
\begin{align}
 &\sum_{k=1}^{A} \mathbb E [\frac{1}{(k^{\gamma}+|S_{t(k,k^\gamma)}|)k^{\gamma}}]
 &>\:& \sum_{k=1}^{A} \frac{1}{(k^{\gamma}+\mathbb E [|S_{10k}|])k^{\gamma}}
 &>\:& \frac{1}{1.23}\sum_{k=1}^{A} \frac{1}{(k^{\gamma}+\sqrt{10k})k^{\gamma}}\notag\\
&  &>\:& \frac{1}{1.23(\sqrt {10} +1)} \sum_{k=1}^{A} \frac{1}{k^{\frac{1}{2}+\gamma}}
  &>\:& \frac{1}{6} \int_{x=1}^{A} \frac{1}{x^{\frac{1}{2}+\gamma}}dx\notag\\
   &&>\:& \frac{1}{6}  [\frac{1}{\frac{1}{2}-\gamma}x^{\frac{1}{2}-\gamma}]_{x=1}^{A}
   &=\:& \Omega (A^{\frac{1}{2}-\gamma})
\end{align}

\Paragraph{Critical matching \textpar{$\CS_{\gamma =1/2}$}}
 For the upper bound, by \cref{lemma_minofexpvariables}, we have:
\begin{equation}
 \mathbb E [\sum_{k=1}^{A} \frac{1}{(\sqrt k+|S_{t(k,\sqrt k)}|)\sqrt k}]
 \: < \:\sum_{k=1}^{A}\frac{1}{k}   \:\leq \: \log A+1 
\end{equation}
where the last inequality follows from the bounds for the harmonic series.
Thus, given the optimal clearing schedule has finite matching cost (lower bounded by $\frac{\log(2)}{\overline \rate}$), the expected matching ratio is smaller than $\frac{\frac{1}{\underline \rate}(\log A+1)}{\log(2)/\overline\rate}$.

For the lower bound note that $\mathbb E [t(k,\sqrt k)]<10k$ by similar arguments as in  \cref{qwerfwe}. Further note that $S_t$ is strictly increasing in $t$. Thus,  with Jensen's inequality ($\frac{1}{x}$ is convex) and \cref{lemma_randomwalkonZdistancefromzero}:
\begin{flalign}
\sum_{k=1}^{A} \mathbb E [\frac{1}{(\sqrt k+|S_{t(k,\sqrt k)}|)\sqrt k}]
	&> \sum_{k=1}^{A} \frac{1}{(\sqrt k+\mathbb E [|S_{10k}|])\sqrt k}
	> \frac{1}{1.23}\sum_{k=1}^{A} \frac{1}{(\sqrt k+\sqrt{10k})\sqrt k}
	\notag\\
	&\geq \frac{1}{6} \sum_{k=1}^{A} \frac{1}{ k}
	> \frac{1}{6}\log A
\end{flalign}
Thus, given the optimal clearing schedule has   $\cost_{\patient}(A)\leq\frac{\pi^2}{6 \underline \rate}$, the expected matching ratio is bounded below by $\frac{\frac{1}{6}\log A}{\pi^2/(6 \underline \rate)}=\frac{\underline \rate }{\pi^2}\cdot \log A$.

\Paragraph{Supercritical matching \textpar{$\CS_{1/2<\gamma\leq1}$}} As above, by Jensen's inequality (since $1/x$ is convex) and \cref{lemma_minofexpvariables} we have:
\begin{equation}
 \mathbb E [ \sum_{k=1}^{A} \frac{1}{(k^\gamma+|S_{t(k,k^\gamma)}|)k^\gamma}]
\: < \: \sum_{k=1}^{A} \mathbb E [\frac{1}{k^\gamma k^\gamma}]
\: = \:\sum_{k=1}^{A}\frac{1}{k^{2\gamma}}\: \rightarrow \: \zeta(2\gamma)
\end{equation}
where $\zeta$ is the Riemann zeta function and is known to converge for $\gamma>\frac{1}{2}$. Given that we are considering a sum with positive summands convergence is from below. Thus, given the optimal clearing schedule has finitematching cost ($\cost_{\patient}(A)\geq \frac{\log(2)}{\overline\rate}$),  the expected matching ratio is bounded from above by $(\overline\rate/\underline{\rate}) \cdot \zeta(2\gamma) / \log2$ for $\frac{1}{2}<\gamma\leq1$.
\end{proof}

%% file: 11_appendix_wait.tex

\begin{proof}[Proof of \cref{THM:WAIT}.]
First note that in order to compare different \aclp{CS} we are interested in the additional waiting time incurred until some number $A$ of couples have been matched. Thus, we consider the waiting time until $\horizon$ for the greedy schedule (the benchmark) and for other schedules the waiting time until $\hat\horizon$
where $\hat\horizon$ is the expected time until under the given schedule the same number of couples have been matched as in the greedy schedule until time $\horizon$.

As in the Proof of \cref{THM:COST} we shall simplify notation by omitting the fact that some of the matching schedules are defined via the ceiling function of functions mapping to $\mathbb R^+$ (e.g., $\ceil{k^\gamma}$). We invite the reader to convince her-or himself that the results are not altered through this simplification.

Let $\tau(k)$ be the moment the $k$-th couple is matched (given a particular \acl{CS}).
We proceed in a case-by-case basis below:

\paragraph{First come, first served \textpar{$\CS_{\FIFO}}$}
It suffices to note that this \aclp{CS} matches players at exactly the same moments as the greedy \aclp{CS}.
The result then follows.

\paragraph{Subcritical and critical matching \textpar{$\CS_{0\leq\gamma \leq 1/2}$}}
We shall study the worst case such \acl{CS} with respect to waiting time. We consider two different parts. In the first part we  wait until at least $\horizon^\gamma$ clients and $\horizon^\gamma$ providers are in the market. The second part then proceeds in the same way as the greedy \acl{CS}, keeping in mind that at all future times $\min\{\nBoys, \nGirls\}$ is exactly $\horizon^\gamma$.
The expected waiting time of the first schedule can be bounded above by the upper bound for the expected time until $\horizon^\gamma$ clients and $\horizon^\gamma$ providers are in the market, that is, $\mathbb E[\tau(5\horizon^\gamma)]=5T^\gamma$ (see \cref{eqhelpeq} in the proof of \cref{THM:COST}) noting that we used the fact that the arrival of agents is governed by a Poisson clock of rate 1. Now, a crude upper bound for the waiting time of the first part of the process is found be assuming that all agents are in the market from the beginning ($\tau=0$), yielding the upper bound $5\horizon^\gamma\cdot 5\horizon^\gamma$.

Note that, the first part of the process takes $\hat{ \horizon}-\horizon$ time. For the remaining second part of the process  the waiting cost is the cost of the greedy schedule ($\frac{2}{3}\horizon^{3/2}$) plus the cost of the \textendash\ in expectation \textendash\ no more than $5T^\gamma$ agents on each side of the market to `remain' for the subsequent periods. Thus the total waiting time is bounded above by:
\begin{equation}
5\horizon^\gamma \cdot 5\horizon^\gamma + \frac{2}{3}\horizon^{3/2}+ 5\horizon^\gamma\cdot \horizon=\Theta (\horizon^{3/2})
\end{equation}
Thus $\beta(\hat{ \horizon}) = (3/2) \, \Theta (\horizon^{3/2}) / \horizon^{3/2} = \Theta(1)$.

\paragraph{Supercritical matching \textpar{$\CS_{1/2<\gamma\leq1}$}}
We first construct a lower bound. Consider the alternative arrival process, where clients and providers alternatingly arrive to the market. Note that for any given \acl{CS} this process incurs lower waiting time. For the \acl{CS} we consider the waiting time of this alternative arrival process is precisely governed by the fact that the $k$-th match takes place when at least $k^\gamma$ players are on the short side of the market.
Further note that $\hat {\horizon} \geq \horizon$. Thus, the waiting time is lower bounded by using the  approximation  by the  Wiener process (by arguments as in \cref{prop:optimalbaseline} and by observing that arrival is governed by a Poisson clock of rate 1):
\begin{equation}
\int_0^{\horizon} 2\tau^\gamma d\tau
	= \frac{2}{1+\gamma}\tau^{1+\gamma}|_0^{\horizon}
	= \Omega (\horizon^{1+\gamma})
\end{equation}

For the upper bound, we construct a \acl{CS} that constitutes an upper bound of the schedule under consideration. For fixed $k$, let $\horizon=\tau(k)$ consider the following \acl{CS}: First wait until there are at least $k^\gamma$ clients and providers in the market, then proceed with the greedy schedule such that at any future point $\min\{clients, providers\}$ in the market is equal to $k^\gamma$. Note that this new schedule has the same total run time as the original schedule, that is, $\hat{ \horizon}$. Further it is evident that the waiting time occurred by the new schedule is greater than the waiting time of the original schedule. By arguments as for $(\CS_{\gamma=0})$ and by the fact that arrival is governed by a Poisson clock of rate 1 we can  upper bound the waiting time by:
 \begin{equation}
\label{eq:upperboundwaithelp}
5\horizon^\gamma \cdot 5\horizon^\gamma
	+ \frac{2}{3}\horizon^{3/2}
	+ 5\horizon^\gamma\cdot \horizon
	= \Theta (\horizon^{1+\gamma})
\end{equation}
since we assumed $\frac{1}{2}<\gamma\leq1$.

The two bounds together show that the waiting time of the originally considered \acl{CS} is $\Theta (\horizon^{1+\gamma})$.Thus $\beta(\hat{ \horizon}) = \frac{\Theta (\horizon^{1+\gamma})}{(2/3) \, \horizon^{3/2}}=\Theta(\horizon^{\gamma-\frac{1}{2}})$.

\paragraph{Patient matching \textpar{$\CS_{\patient}$}}
First note that for the patient schedule $\hat{ \horizon}=\horizon$. The expected waiting time for the patient schedule until time $\horizon$ is given by
\begin{equation}
\exof*{\int_0^{\horizon} \nBoys(\tau) + \nGirls(\tau) \dd\tau }
	= \int_0^{\horizon} \exof{\nBoys(\tau)+\nGirls(\tau)} \dd\tau
\end{equation}
where the latter equality holds  by Tonelli's theorem (by noting that $\nBoys(\tau)+\nGirls(\tau)$ is non-negative). The expectation is with respect to the number of clients and providers and with respect to the arrival times of the agents (governed by a Poisson clock). Again by Tonelli's theorem we can consider the case where the expectation with respect to the arrival times is taken first. Then by the fact  that the arrival of agents is assumed to follow a Poisson clock of rate 1 we have: 
\begin{equation}
\int_0^{\horizon} \mathbb E [ \nBoys(\tau)+\nGirls(\tau)]d\tau
	= \int_0^{\horizon} \floor{\tau} d\tau= \Theta(\horizon^2)
\end{equation}
 Thus $\beta(\hat{ \horizon})=\frac{\Theta( \horizon^2)}{\frac{2}{3}\horizon^{3/2}}=\Theta (\sqrt\horizon)$. 
\end{proof}

%% file: 12_appendix_other.tex
\begin{proof}[\textbf{Proof of omitted approximation in Proof of \cref{prop:socialplanner}}]
We begin by approximating the two sums in \cref{eqeqeqeqeq}, \ie
\begin{equation}
\sum_{k=1}^{A} \frac{1}{f(k)^{2}}
	= \Theta \parens*{ \frac{1}{A^{3/2}}\sum_{k=1}^A  f(k) }
\end{equation}
Recalling that $f$ is assumed non-decreasing for large $k$, the summand on the \acl{LHS} is decreasing and
\begin{equation}
\int_0^A \frac{1}{f(x)^{2}} \dd x
	\geq \sum_{k=1}^{A} \frac{1}{f(k)^{2}}
	\geq \int_1^{A+1} \frac{\dd x}{f(x)^2}
\end{equation}
Considering the meaning of $f(k)$ it is without loss of generality to define $f(x)=1$ for $x\in [0,1)$ since the summand $\frac{1}{f(k)^{2}}$ remains decreasing. Thus the absolute difference between the two bounds is bounded above by:
\begin{equation}
\abs*{\int_0^A \frac{1}{f(x)^{2}} \dd x - \int_1^{A+1} \frac{1}{f(x)^{2}} \dd x}
	= \abs*{\int_0^1 \frac{1}{f(x)^{2}} \dd x - \int_A^{A+1} \frac{1}{f(x)^{2}} \dd x}
	\leq 1
\end{equation}
It follows that 
\begin{equation}
\sum_{k=1}^{A} \frac{1}{f(k)^{2}}
	=\Theta\parens*{\int_1^{A+1} \frac{1}{f(x)^{2}}\dd x}
\end{equation}

Next consider the \acl{RHS} of \eqref{eqeqeqeqeq}.
The summand is increasing, so we get:
\begin{equation}
\frac{1}{A^{3/2}} \int_0^A f(x) \dd x
	\leq \frac{1}{A^{3/2}} \sum_{k=1}^{A} f(k)
	\leq \frac{1}{A^{3/2}}\int_1^{A+1} f(x) \dd x
\end{equation}
Now note that $f(x)<x$ must hold. Thus the absolute difference between the two bounds is bounded above by:
\begin{flalign}
\frac{1}{A^{3/2}} \abs*{\int_0^A f(x) \dd x  - \int_1^{A+1} f(x) \dd x}
	&= \frac{1}{A^{3/2}} \abs*{\int_A^{A+1} f(x) dx  - \int_0^{1} f(x) \dd x}
	\notag\\
&	\leq \frac{A}{A^{3/2}}
	= \bigoh(1).
\end{flalign}
It follows that  
\begin{equation}
\frac{1}{A^{3/2}} \sum_{k=1}^{A} f(k)
	= \Theta\parens*{\frac{1}{A^{3/2}}\int_1^{A+1} f(x) \dd x}
\end{equation}
With above approximations it follows that \cref{eqeqeqeqeq} holds if and only if the following equation holds:
\begin{equation}
\int_1^{A} \frac{1}{f(x)^{2}} \dd x
	= \Theta\parens*{\frac{1}{A^{3/2}}\int_1^{A} f(x) \dd x}
\end{equation}
as claimed.
\end{proof}